\documentclass[10pt, conference, letterpaper]{IEEEtran}

\IEEEoverridecommandlockouts
\usepackage{cite}
\usepackage{amsmath,amssymb,amsfonts}
\usepackage{algorithmic,verbatim}
\usepackage{graphicx}
\usepackage{textcomp}
\usepackage{xcolor}
\def\BibTeX{{\rm B\kern-.05em{\sc i\kern-.025em b}\kern-.08em
    T\kern-.1667em\lower.7ex\hbox{E}\kern-.125emX}}

\title{Kelly Cache Networks}

\usepackage{amssymb} 
\usepackage{url}
\usepackage{mathtools}
\usepackage[export]{adjustbox}
\usepackage{subfig}
\usepackage{amsthm} 
\usepackage{cuted}
\usepackage{graphicx}
\DeclareGraphicsExtensions{.pdf,.png,.jpg}
\usepackage{bm}
\usepackage{algorithm}
\usepackage{algorithmic}
\usepackage{booktabs} 

\newtheorem{theorem}{Theorem}
\newtheorem{lemma}{Lemma}

\newtheorem{cor}{Corollary}
\newtheorem{assumption}{Assumption}

\newcommand{\demand}{\mathcal{R}}
\newcommand{\catalog}{\mathcal{C}}
\newcommand{\servers}{\mathcal{S}}
\newcommand{\domain}{\mathcal{D}}

\newcommand{\conv}{\mathop{\mathtt{conv}}}
\newcommand{\supp}{\mathop{\mathtt{supp}}}
\newcommand{\reals}{\mathbb{R}}
\newcommand{\naturals}{\mathbb{N}}
\newcommand{\prob}{\mathbf{P}}
\newcommand{\expect}{\mathbb{E}}

\newcounter{packednmbr}

\author{
\IEEEauthorblockN{Milad Mahdian, Armin Moharrer, Stratis Ioannidis, and Edmund Yeh}
\IEEEauthorblockA{Electrical and Computer Engineering,
Northeastern University,
Boston, MA, USA  \\
 \{mmahdian,amoharrer,ioannidis,eyeh\}@ece.neu.edu}
}

\newcommand{\techrep}[2]{#1}

\usepackage[font=scriptsize]{caption}

\begin{document}
\maketitle

\begin{abstract}
	We study networks of M/M/1 queues in which nodes act as caches that store objects. Exogenous requests for objects  are routed towards nodes that store them; as a result, object traffic in the network is determined not only by demand but, crucially, by where objects are cached. We determine how to place objects in caches to attain a certain design objective, such as, e.g., minimizing network congestion  or retrieval delays.  We show that for a broad class of objectives, including minimizing both the expected network delay and the sum of network queue lengths, this optimization problem can be cast as an NP-hard submodular maximization problem. We show that so-called \emph{continuous greedy} algorithm~\cite{calinescu2011} attains a ratio arbitrarily close to $1-1/e\approx 0.63$ using a deterministic estimation via a power series; this drastically reduces execution time over prior art, which resorts to sampling. Finally, we show that our results generalize, beyond M/M/1 queues, to networks of M/M/$k$ and symmetric M/D/1 queues.
\end{abstract}

\techrep{}{
}

\captionsetup[algorithm]{font=scriptsize}
\captionsetup[table]{font=scriptsize}

	\section{Introduction}\label{sec:introduction}
	Kelly networks \cite{kelly} are multi-class networks of queues capturing a broad array of queue service disciplines, including FIFO, LIFO, and processor sharing. Both Kelly networks and their  generalizations (including networks of quasi-reversible and symmetric queues) are well studied and classic topics \cite{kelly,gallager-stochastic,nelson,chen2013fundamentals}. One of their most appealing properties is that their steady-state distributions have a product-form: as a result, steady state properties such as expected queue sizes, packet delays,  and server occupancy rates have closed-form formulas 
  as functions of, e.g., routing and scheduling policies.

In this paper, we  consider Kelly networks in which nodes are equipped with caches, i.e., storage devices of finite capacity, which can be used to store objects. Exogenous  requests for objects  are routed towards nodes that store them; upon reaching a node that stores the requested object, a response packet containing the object is routed towards the request source. As a result, object traffic in the network is determined not only by the demand but, crucially, by where objects are cached.
This abstract setting is motivated by--and can be used to model--various networking applications involving the placement and transmission of content. This includes information centric networks ~\cite{jacobson2009networking,yeh2014vip,ioannidis2016adaptive}, content delivery networks~\cite{borst2010distributed,dehghan2014complexity}, web-caches~\cite{laoutaris2004meta,che2002hierarchical,zhou2004second},  wireless/femtocell networks \cite{shanmugam2013femtocaching,naveen2015interaction,poularakis2013approximation}, and 
 peer-to-peer networks~\cite{lv2002search,cohen2002replication}, to name a few.

In many of these applications, determining the \emph{object placement}, i.e., how to place objects in network caches, is a decision that can be made by the network designer in response to object popularity and demand. To that end, we are interested in determining how to place objects in caches so that  traffic attains a design objective such as, e.g., minimizing delay.

We make the following contributions.
First, we study the problem of optimizing the placement of objects in caches in  Kelly cache networks of M/M/1 queues, with the objective of minimizing a cost function of the system state. We show that, for a broad class of cost functions, including packet delay, system size, and server occupancy rate, \emph{this optimization amounts to a submodular maximization problem with matroid constraints}. This result applies to general Kelly networks with fixed service rates; in particular, it holds for FIFO, LIFO, and processor sharing  disciplines at each queue.

The so-called continuous greedy algorithm \cite{calinescu2011} attains a $1-1/e$ approximation for this NP-hard problem. However, it does so by computing an expectation over a random variable with exponential support  via randomized sampling. The number of samples required to attain the $1-1/e$ approximation guarantee can be prohibitively large in realistic settings. Our second contribution is to show that, for Kelly networks of M/M/1 queues, \emph{this randomization can be entirely avoided}: a closed-form solution can be computed using the Taylor expansion of our problem's objective. To the best of our knowledge, we are the first to identify a submodular maximization problem that exhibits this structure, and to exploit it to eschew sampling.
%
Finally, we extend our results to \emph{networks of M/M/$k$  and symmetric M/D/1  queues}, and 
prove  a negative result: submodularity does \emph{not} arise in networks of M/M/1/$k$ queues.
 We extensively evaluate our proposed algorithms over several synthetic and real-life topologies. 

The remainder of our paper is organized as follows. We review related work in Sec.~\ref{sec:related}. We present our mathematical model of a Kelly cache network in Sec.~\ref{sec:model}, and our results on submodularity and the continuous-greedy algorithm in networks of M/M/1 queues in Sections~\ref{sec:subm} and~\ref{sec:contgreed}, respectively. Our extensions are described in Sec.~\ref{sec:beyond}; our numerical evaluation is in Sec.~\ref{sec:experiments}. Finally, we conclude in Sec.~\ref{sec:conclusions}. \techrep{}{For brevity, we omit detailed proofs, which can be found in \cite{techrep}.}

	\section{Related Work}\label{sec:related}

	Our approach is closest to, and inspired by, recent work by Shanmugam et al. \cite{femtocaching} and Ioannidis and Yeh  ~\cite{ioannidis2016adaptive}. 
 Ioannidis and Yeh consider a setting very similar to ours but without  queuing: edges are assigned a fixed weight, and the objective is a linear function of incoming traffic scaled by these weights. This can be seen as a special case of our model, namely, one where edge costs are linear (see also Sec.~\ref{sec:mincost}). Shanmugam et al.~\cite{femtocaching} study a similar optimization problem, restricted to the context of femtocaching. The authors show that this is an NP-hard, submodular maximization problem with matroid constraints. They provide a $1-1/e$ approximation algorithm based on a technique by Ageev and Sviridenko \cite{ageev2004pipage}: this involves maximizing a concave relaxation of the original objective, and rounding via pipage-rounding\cite{ageev2004pipage}. Ioannidis and Yeh show that the same approximation technique applies to more general cache networks with linear edge costs. They also provide a distributed, adaptive algorithm that attains an   $1-1/e$ approximation. The same authors  extend this framework to jointly optimize both caching and routing decisions \cite{ioannidis2017jointly}.

Our work can be seen as an extension of~\cite{ioannidis2016adaptive,femtocaching}, in that it incorporates queuing in the cache network. In contrast to both \cite{ioannidis2016adaptive} and \cite{femtocaching} however, costs like delay or queue sizes are highly non-linear in the presence of queuing. From a technical standpoint, this departure from linearity requires us to employ significantly different  optimization methods than the ones in \cite{ioannidis2016adaptive,femtocaching}. In particular, our objective \emph{does not admit a concave relaxation} and, consequently, the technique by Ageev and Sviridenko \cite{ageev2004pipage} used in \cite{ioannidis2016adaptive,femtocaching} \emph{does not apply}.  Instead, we must solve a non-convex optimization problem   directly  (c.f.~Eq.~\eqref{eq:nonconv}) using the so-called continuous-greedy algorithm.

 Several papers have studied the  cache optimization problems under restricted topologies \cite{baev2008approximation,bartal1995competitive,fleischer2006tight,applegate2010optimal,borst2010distributed}. These works model the network as a bipartite graph:  nodes generating requests connect directly to caches in a single hop.  The resulting algorithms do not readily generalize to arbitrary topologies. In general, the approximation technique of Ageev and Sviridenko~\cite{ageev2004pipage} applies to this bipartite setting, and additional approximation algorithms have been devised for several variants~\cite{baev2008approximation,bartal1995competitive,fleischer2006tight, borst2010distributed}. We differ   by (a)  considering a multi-hop setting, and (b) introducing queuing, which none of the above works considers.

Submodular function maximization subject to matroid constraints appears in many important problems in combinatorial optimization; 
for a brief review of the topic and applications, see \cite{krause2012} and \cite{goundan2007revisiting}, respectively. Nemhauser et al. \cite{greedy2} show that the greedy algorithm produces a solution  within 1/2 of the optimal.   Vondr\'ak \cite{vondrak2008optimal} and Calinescu et al. \cite{calinescu2011} show that the continuous-greedy algorithm produces a solution within $(1-1/e)$ of the optimal in polynomial time, which cannot be further improved \cite{nemhauser1978best}. 
In the general case, the continuous-greedy algorithm requires sampling to estimate the gradient of the so-called multilinear relaxation of the objective (see Sec.~\ref{sec:contgreed}). One of our main contributions is to show that \textsc{MaxCG}, the optimization problem we study here, exhibits additional structure: we use this  to construct a sampling-free estimator of the gradient via a power-series or Taylor expansion. To the best of our knowledge, we are the first to use such an expansion to eschew sampling; this technique may apply to submodular maximization problems beyond \textsc{MaxCG}.

	\section{Model}\label{sec:model}
	Motivated by applications such as ICNs \cite{jacobson2009networking}, CDNs \cite{borst2010distributed,dehghan2014complexity}, and peer-to-peer networks \cite{lv2002search}, we introduce Kelly cache networks. In contrast to classic Kelly networks, each node is associated with a cache of finite storage capacity. Exogenous traffic consisting of \emph{requests} is routed towards nodes that store objects; upon reaching a node that stores the requested object, a \emph{response} packet containing the object is routed towards the node that generated the request. As a result, content traffic in the network is determined not only by demand but, crucially, by how contents are cached. 
 \techrep{For completeness, we  review classic Kelly networks in Appendix~\ref{app:classickelly}.}{} An illustration highlighting the differences between Kelly cache networks, introduced below, and classic Kelly networks, can be found in Fig.~\ref{fig:kelly-classic}.
\begin{figure}
\includegraphics[width=0.49\columnwidth]{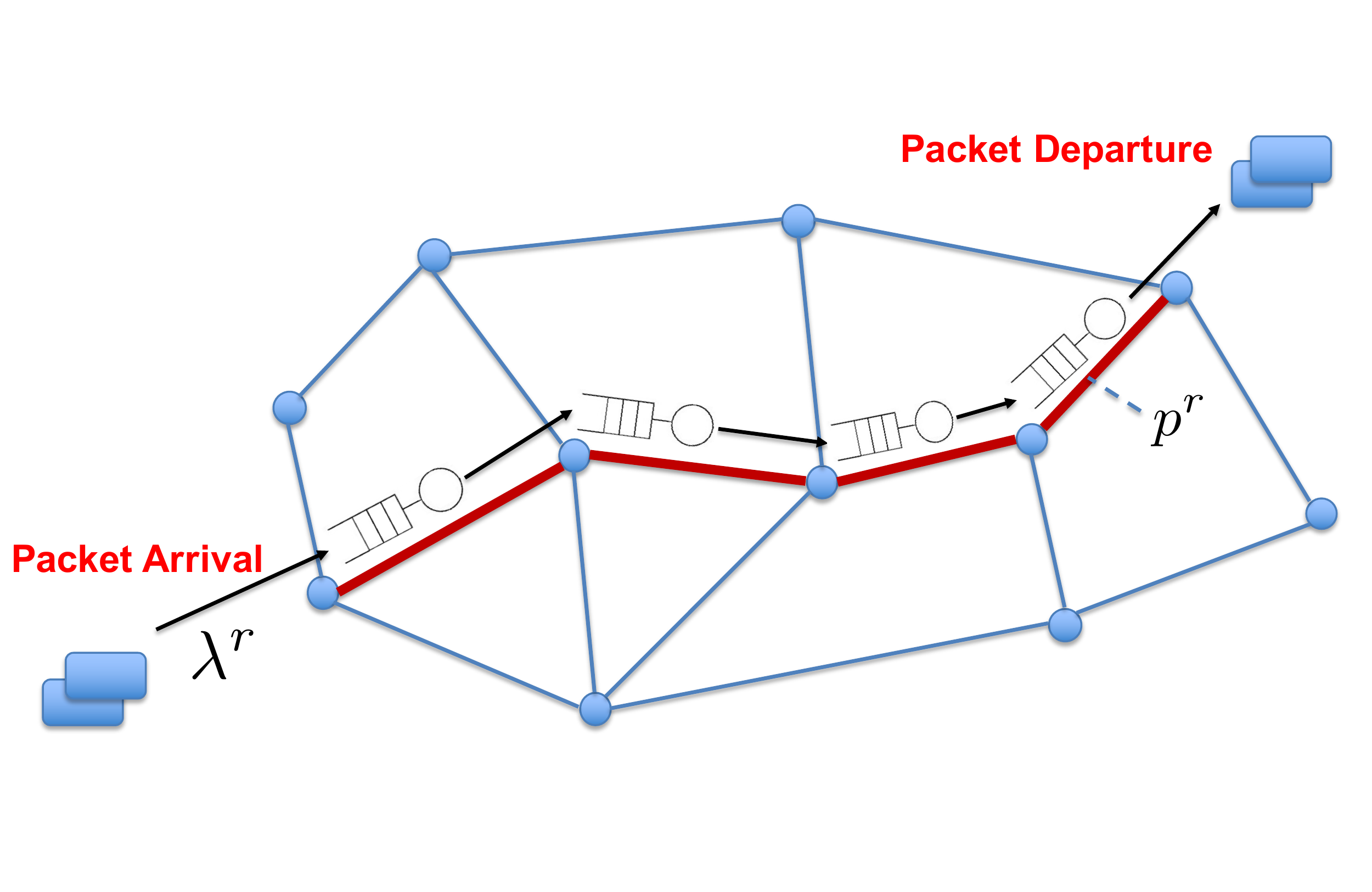}\includegraphics[width=0.49\columnwidth]{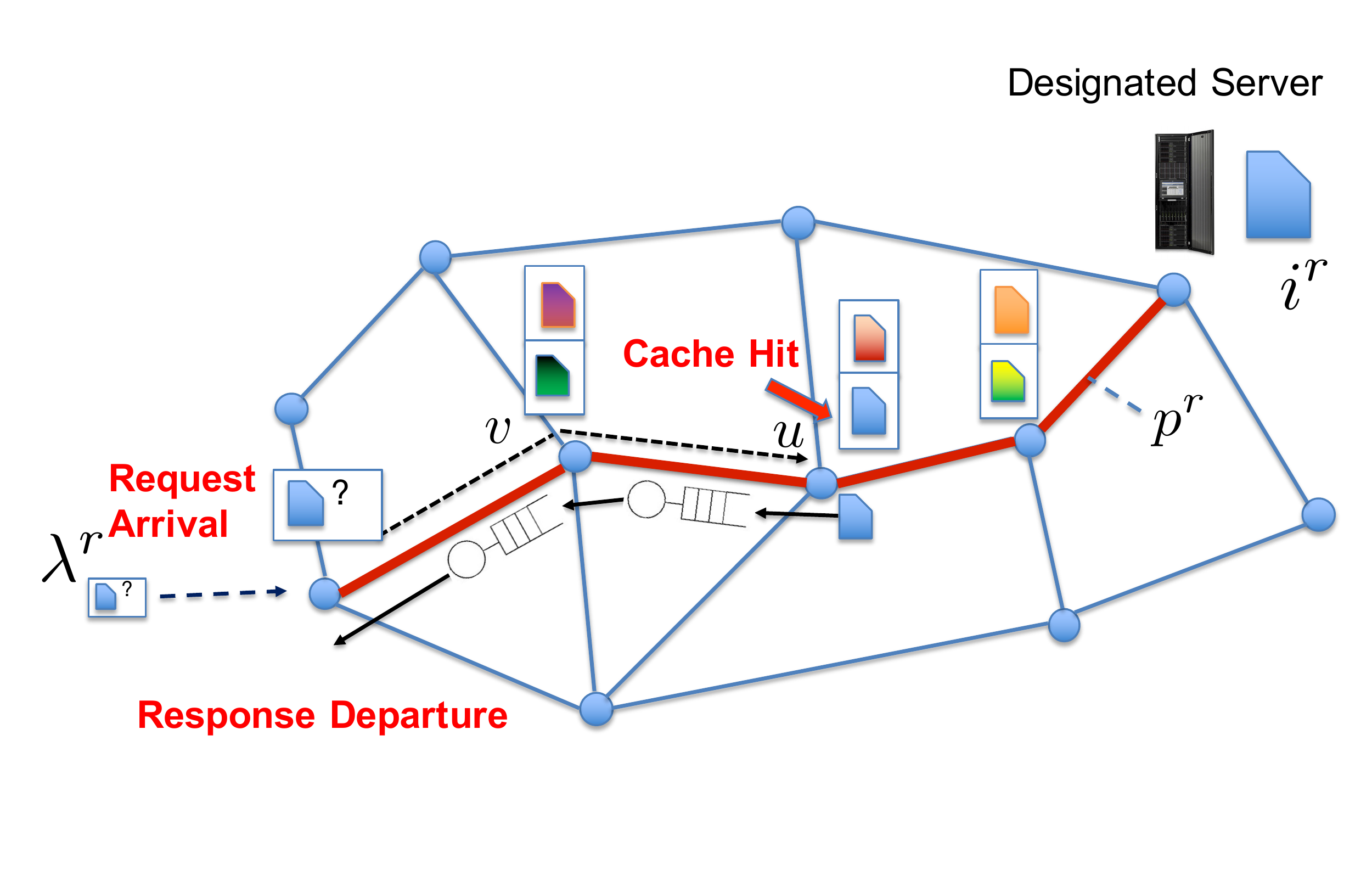}\\
\begin{scriptsize} \hspace*{\stretch{1.3}}{(a) Kelly Network}\hspace*{\stretch{2}}
{(b) Kelly Cache Network}\hspace*{\stretch{1}}\end{scriptsize}\\
\caption{(a) Example of a Kelly network. Packets of class $r$ enter the network with rate $\lambda^r$,  are routed through consecutive queues over path $p^r$, and subsequently exit the network. (b) Example of a Kelly cache network. Each node $v\in V$ is equipped with a cache of capacity $c_v$. Exogenous requests of type $r$ for object $i^r$  enter the network and are routed  over a predetermined path $p^r$ towards the designated server storing $i^r$. Upon reaching an intermediate node $u$ storing the requested object $i^r$, a response packet containing the object is generated. The response is then forwarded towards the request's source in the reverse direction on path $p^r$. Request packets are of negligible size compared to response messages; as a result, we ignore request traffic and focus on queuing due to response traffic alone.} 
\label{fig:kelly-classic}
\end{figure}

Although we describe Kelly cache networks in terms of FIFO M/M/1 queues, the product form distribution  (c.f.~\eqref{steadystate}) arises for many different service principles beyond FIFO (c.f.~Section 3.1 of \cite{kelly}) 
including Last-In First-Out (LIFO) and processor sharing. All results we present extend to these service disciplines; we discuss 
more extensions in Sec.~\ref{sec:beyond}.

\subsection{Kelly Cache Networks} 

\noindent\textbf{Graphs and Paths.} We use the notation $G(V,E)$ for a directed graph $G$ with nodes $V$ and edges $E\subseteq V\times V$. A directed graph is called \emph{symmetric} or \emph{bidirectional} if $(u,v)\in E$ if and only if $(v,u)\in E$. A \emph{path} $p$ is a sequence of adjacent nodes, i.e.,  $p=p_1,p_2,\ldots, p_K$ where $(p_{k},p_{k+1})\in E$, for all $1\leq i<K\equiv|p|$. A path is \emph{simple} if it contains no loops (i.e., each node appears once). We use the notation $v\in p$, where $v\in V$, to indicate that node $v$ appears in the path, and $e\in p$, where $e=(u,v)\in E$, to indicate that nodes $u$,$v$ are two consecutive (and, therefore, adjacent) nodes in $p$.  For $v\in p$, where $p$ is simple, we denote by $k_p(v)\in \{1,\ldots,|p|\}$ the position of node $v\in V$ in $p$, i.e., $k_p(v)=k$ if $p_k=v$.


\noindent\textbf{Network Definition.} Formally,  we consider a Kelly network of M/M/1  FIFO queues, represented by a symmetric directed graph $G(V,E)$. As in classic Kelly networks, each edge $e\in E$ is associated with an M/M/1 queue with service rate $\mu_e$\footnote{We associate queues with edges for concreteness. Alternatively, queues can be associated with nodes, or both nodes and edges; all such representations lead to product form distributions \eqref{steadystate}, and all our results extend to these cases.}. In addition, each node has a cache that stores  objects of equal size from a set $\catalog$,  the \emph{object catalog}.  Each node $v\in V$ may store at most $c_v\in \naturals$  objects from $\catalog$ in its cache.  Hence, if $x_{vi} \in \{0,1\}$ is a binary variable indicating whether node $v\in V$ is storing object $i\in \catalog$, then 	
$		\sum_{i\in\catalog} x_{vi} \leq c_v,$ for all $v\in V. $
We refer to $\mathbf{x}=[x_{vi}]_{v\in V, i\in \catalog} \in \{0,1\}^{|V||\catalog|}$ as the \emph{global placement} or, simply, \emph{placement} vector. We denote by
\begin{align}
\domain = \left\{\mathbf{x}\in \{0,1\}^{|V| |\catalog|}: \textstyle\sum_{i\in\catalog} x_{vi} \leq c_v, \forall v\in V \right \}, \label{domain}
\end{align}
the set of \emph{feasible} placements that satisfy the storage capacity constraints. 
We assume that for every object $i\in\catalog$, there exists a set of nodes $\servers_i\subseteq V$ that \emph{permanently store} $i$. We refer to nodes in $\servers_i$ as \emph{designated servers} for $i\in \catalog$. We assume that designated servers store $i$ in permanent storage \emph{outside} their cache. Put differently, the aggregate storage capacity of a node is $c_v'=c_v+|\{i:v\in \servers_i\}|$, but only the non-designated slots $c_v$ are part of the system's design.

	
\noindent\textbf{Object Requests and Responses.}
Traffic in the cache network consists of two types of packets: \emph{requests} and \emph{responses}, as shown in  Fig.~\ref{fig:kelly-classic}(b).
 Requests for an object are always routed towards one of its designated servers, ensuring  that every request is satisfied. However, requests may terminate early: upon reaching any node that caches the requested object,  the latter generates a response carrying the object. This is forwarded  towards the request's source, following the same path as the request, in reverse. Consistent with prior literature \cite{ioannidis2016adaptive,ioannidis2017jointly}, 
 we treat request  traffic as negligible when compared to response traffic, which carries objects,  and henceforth focus only on queues bearing response  traffic.

Formally, a request and its corresponding response are fully characterized by (a)  the object  being requested, and (b) the path  that the request  follows. That is, for the set of requests $\demand$, a request $r\in \demand$ is determined by a pair $(i^r,p^r)$, where $i^r\in \catalog$ is the object being requested and $p^r$ is the path the request follows. 
Each request $r$ is associated with a corresponding Poisson arrival process with rate $\lambda^r\geq 0$, independent of other arrivals and service times.  We denote the vector of arrival rates by
$ \bm{\lambda} = [\lambda^r]_{r\in \demand} \in\reals_+^{|\demand|}.$
%
For all $r\in \demand$, we assume that the path $p^r$ is well-routed \cite{ioannidis2016adaptive}, that is: (a) path $p^r$ is simple, (b) the terminal node of the path is a designated server, i.e., a node in $\servers_i$, and (c) no other intermediate node in $p^r$ is a designated server. As a result, requests are always served, and response packets (carrying objects) always follow a sub-path of $p^r$ in reverse towards the request source (namely, $p^r_1$). 

	
	
	

\noindent\textbf{Steady State Distribution.}	
Given an object placement $\mathbf{x}\in \domain$, the resulting system is  a multi-class Kelly network, with packet classes determined by the request set $\demand$. This is a Markov process over the state space determined by queue contents. In particular, let $n_e^r$ be the number of packets of class $r\in \demand$ in queue $e\in E$, and $n_e = \sum_{r\in \demand} n_e^r$ be the total queue size. The state of a queue $\mathbf{n}_{e}\in \demand^{n_{e}}$, $e\in E$, is the vector of length $n_{e}$ representing the class of each packet in each position of the queue. The \emph{system state} is then given by $\mathbf{n}=[\mathbf{n}_{e}]_{e\in E}$;  we denote by $\Omega$ the state space of this Markov process.

In contrast to classic Kelly networks, network traffic and, in particular, the load on each queue, 
 depend on placement $\mathbf{x}$. Indeed,  if $(v,u)\in p^r$ for $r\in \demand$, the arrival rate of responses of class $r\in \demand$ in queue $(u,v)\in E$ is:
\begin{align}\textstyle\lambda_{(u,v)}^r(\mathbf{x},\bm{\lambda}) =  \lambda^r \prod\limits_{k'=1}^{k_{p^r}(v)}(1-x_{p_{k'}^ri^r}),\label{lamr}\quad \text{for}~(v,u)\in p^r, \end{align}
i.e., responses to requests of class $r$ pass through edge $(u,v)\in E$ if and only if no node preceding $u$ in the path $p^r$ stores object $i^r$--see also Fig.~\ref{fig:kelly-classic}(b). As $\mu_{(u,v)}$ is the service rate of the queue in $(u,v)\in E$, 
the  load on edge $(u,v)\in E$ is:
	 		\begin{align}
	 		\textstyle\rho_{(u,v)}(\mathbf{x},\bm{\lambda}) = \frac{1}{\mu_{(u,v)}} \textstyle\sum_{r\in\demand: (v,u)\in p^r}\lambda^r_{(u,v)}(\mathbf{x},\bm{\lambda}). 
\label{eq:rho_X}
	 		\end{align} 
The Markov process $\{\mathbf{n}(t);t\geq 0\}_{t\geq 0}$ is positive recurrent when $\rho_{(u,v)}(\mathbf{x},\bm{\lambda})< 1$, for all $(u,v)\in E$ \cite{kelly,datanetworks}. Then, the steady-state distribution has a \emph{product form}, i.e.:
\begin{align}
\textstyle\pi(\mathbf{n}) =  \prod_{e\in E} \pi_e(\mathbf{n}_e), \quad \mathbf{n}\in \Omega,\label{steadystate}
\end{align}
where
$\textstyle\pi_e(\mathbf{n}_e) = (1-\rho_e(\mathbf{x},\bm{\lambda})) \prod_{r\in \demand:e\in p^r} \left(\frac{\lambda^r_e(\mathbf{x},\bm{\lambda}) }{\mu_e}\right)^{n_e^r},$
and $\lambda_e^r(\mathbf{x},\bm{\lambda})$, $\rho_e(\mathbf{x},\bm{\lambda})$ are given by \eqref{lamr}, \eqref{eq:rho_X}, respectively.

	
	\noindent\textbf{Stability Region.} Given a placement $\mathbf{x}\in \mathcal{D}$,
 a vector of arrival rates $\bm{\lambda} = [\lambda^r]_{r\in \demand}$ yields a stable (i.e., positive recurrent) system  
	if and only if $\bm{\lambda} \in \Lambda_{\mathbf{x}}$, where
	\begin{equation}
		\Lambda_{\mathbf{x}} := \{\bm{\lambda}: \bm{\lambda} \geq 0: \rho_e(\mathbf{x},\bm{\lambda})< 1, \forall e\in E \} \subset \reals_+^{|\demand|} ,
	\end{equation}
	where loads $\rho_e$, $e\in E$, are given by 
\eqref{eq:rho_X}. 
Conversely, given a vector $\bm{\lambda} \in \reals_+^{|\demand|}$,
\begin{align}
\domain_{\bm{\lambda}} = \{ \mathbf{x}\in \domain:  \rho_e(\mathbf{x},\bm{\lambda})< 1, \forall e\in E  \} \subseteq \domain \label{eq:stabledom}
\end{align}
is the set of feasible placements under which the system is stable. It is easy to confirm that, by the monotonicity of $\rho_e$ w.r.t.~$\mathbf{x}$, 
if $\mathbf{x}\in \domain_{\bm{\lambda}}$ and $\mathbf{x}'\geq \mathbf{x},$  then $\mathbf{x}'\in \domain_{\bm{\lambda}}$, 
where the vector inequality $\mathbf{x}'\geq \mathbf{x}$ is component-wise.
In particular, if $\mathbf{0}\in D_{\bm{\lambda}}$ (i.e., the system is stable without caching), then $\domain_{\bm{\lambda}} = \domain$.

	\subsection{Cache Optimization}\label{sec:mincost}

	Given a Kelly cache network represented by graph $G(V,E)$, service rates $\mu_e$, $e\in E$, storage capacities $c_v$, $v\in V$, a set of requests $\demand$, and arrival rates $\lambda_r$, for $r\in \demand$, we wish to determine placements $\mathbf{x}\in \domain$ that optimize a certain design objective. In particular, we seek placements that are solutions to optimization problems of the following form:\\
	\begin{subequations}\label{opt:mincost}
{\hspace*{\stretch{1}} \textsc{MinCost} \hspace*{\stretch{1}} } 
	\vspace{-0.8em}	
		\begin{align}
	\textstyle \text{Minimize:}&\quad C(\mathbf{x}) = \textstyle\sum_{e\in E} C_e(\rho_e(\mathbf{x},\bm{\lambda})),\label{eq:obj}\\   
			\text{subj.~to:}&\quad \mathbf{x} \in\domain_{\bm{\lambda}} ,
	        \end{align} 
	\end{subequations}
where $C_e:[0,1)\to \reals_+$, $e\in E$, are positive \emph{cost} functions, $\rho_e:\domain\times \reals_+^{|\demand|}\to \reals_+$ is the load on edge $e$, given by \eqref{eq:rho_X}, and $\domain_{\bm{\lambda}}$ is the set of feasible placements that ensure stability, given by \eqref{eq:stabledom}. 
We make the following standing assumption on the cost functions appearing in \textsc{MinCost}:
\begin{assumption} For all $e\in E$, functions $C_e:[0,1)\to \reals_+$ are convex and non-decreasing on $[0,1)$. \label{as:conv}
\end{assumption}
Assumption \ref{as:conv} is natural; indeed it holds for many cost functions that often arise in practice. We list several examples:

\noindent\textbf{Example 1. Queue Size:} Under steady-state distribution~\eqref{steadystate}, the expected number of packets in queue $e\in E$ is given by
$\expect[n_e] = C_e(\rho_e) = \frac{\rho_e}{1-\rho_e},$
which is indeed convex and non-decreasing for $\rho_e\in[0,1)$. Hence, the expected total number of packets in the system in steady state can indeed be written as the sum of such functions.

\noindent\textbf{Example 2. Delay:} From Little's Theorem \cite{datanetworks}, the expected delay experienced by a packet in the system is
$\expect[T] = \frac{1}{\|\bm{\lambda}\|_1} \sum_{e\in E} \expect[n_e],$
where $\|\bm{\lambda}\|_1=\sum_{r\in\demand} \lambda^r$ is the total arrival rate, and $\expect[n_e]$ is the expected size of each queue. Thus, the expected delay can also be written as the sum of functions that satisfy Assumption~\ref{as:conv}. We note that the same is true for the sum of the expected delays per queue $e\in E$, as the latter are given by
$\expect[T_e] = \frac{1}{\lambda_e} \expect[n_e] = \frac{1}{\mu_e(1-\rho_e)},$
which are also convex and non-decreasing in $\rho_e$.

\noindent\textbf{Example 3. Queuing Probability/Load per Edge:} In a FIFO queue, the \emph{queuing probability} is the probability of arriving in a system where the server is busy; this is given by $C_e(\rho_e)=\rho_e=\lambda_e/\mu_e$, which is again non-decreasing and convex.  
This is also, of course, the load per edge. By treating $1/\mu_e$ as the weight of edge $e\in E$, this setting recovers the objective  of \cite{ioannidis2016adaptive} as a special case of our model.

\noindent\textbf{Example 4. Monotone Separable Costs:} More generally, consider a state-dependent cost function $c:\Omega \to \reals_+$ that satisfies the following three properties: (1) it is separable across queues, (2) it depends only on queue sizes $n_e$, and (3) it is non-decreasing w.r.t.~these queue sizes. Formally, $c(\mathbf{n}) = \sum_{e\in E} c_e(n_e),$
where $c_e:\naturals\to \reals_+$, $e\in E$, are non-decreasing functions of the queue sizes. Then,
the steady state cost under distribution \eqref{steadystate} has precisely form \eqref{eq:obj} with convex costs, i.e.,
$\expect[c(\mathbf{n})] = \sum_{e\in E} C_e(\rho_e)$
where $C_e:[0,1)\to\reals_+$ satisfy Assumption~\ref{as:conv}. 
 This follows from the fact that: \begin{align}
\!\!\!C_e(\rho_e) \!\equiv\! \expect[c_e(\mathbf{n})]\!=\! \textstyle c_e(0)+ \sum_{n=0}^{\infty}(c_e(n\!+\!1)\!-\!c_e(n)) \rho_e^n. \!\! \label{eq:power}
\end{align}
\techrep{The proof is in Appendix~\ref{app:convexity}.}{}

In summary, \textsc{MinCost} captures many natural cost objectives, while Assumption~\ref{as:conv} holds for \emph{any} monotonically increasing cost function that depends only on queue sizes.

\begin{table}[!t]
\begin{scriptsize}
\caption{Notation Summary}
\begin{tabular}{p{0.06\textwidth}p{0.38\textwidth}}
\hline
\multicolumn{2}{c}{\textbf{Kelly Cache Networks}}\\
\hline
$G(V,E)$ & Network graph, with nodes $V$ and edges $E$\\
$k_p(v)$ & position of node $v$ in path $p$\\
$\mu_{(u,v)}$ & Service rate of edge $(u,v)\in E$\\
$\demand$ & Set of classes/types of requests \\
$\lambda^r$ & Arrival rate of class $r\in \demand$ \\
$p^r$ & Path followed by  class $r\in \demand$ \\
$i^r$ & Object  requested by class $r\in \demand$ \\
$\catalog$ & Item catalog\\
$\mathcal{S}_i$ & Set of designated servers of $i\in \catalog$\\
$c_v$ & Cache capacity at node $v\in V$\\
$x_{vi}$ & Variable indicating whether $v\in V$ stores $i\in\catalog$ \\
$\mathbf{x}$ & Placement vector of $x_{vi}$s, in $\{0,1\}^{|V||\catalog|}$\\ 
$\bm{\lambda}$ & Vector of arrival rates $\lambda^r$, $r\in \demand$\\
$\lambda_e^r$ & Arrival rate of class $r$ responses over edge $e\in E$\\
$\rho_e$ & Load on edge $e\in E$\\
$\Omega$ & State space\\
$\mathbf{n}$ & Global state vector in $\Omega$\\
$\pi(\mathbf{n})$ & Steady-state distribution of $\mathbf{n}\in \Omega$\\
$\mathbf{n}_e$ & State vector of queue at edge $e\in E$\\
$\pi_e(\mathbf{n}_e)$ & Marginal of steady-state distribution of queue $\mathbf{n}_e$\\
$n_e$ & Size of queue at edge $e\in E$\\
\hline
\multicolumn{2}{c}{\textbf{Cache Optimization}}\\
\hline
$C$ & Global Cost function\\
$C_e$ & Cost function of edge $e\in E$\\
$\domain$ & Set of placements $\mathbf{x}$ satisfying capacity constraints\\
$\mathbf{x}_0$ & A feasible placement in $\domain$\\
$F(\mathbf{x})$ & Caching gain of placement $\mathbf{x}$ over $\mathbf{x}_0$\\
$y_{vi}$ & Probability that $v\in V$ stores $i\in\catalog$ \\
$\mathbf{y}$ & Vector of marginal probabilities $y_{vi}$, in $\{0,1\}^{|V||\catalog|}$\\ 
$G(\mathbf{y})$ & Multilinear extension under marginals $\mathbf{y}$\\
$\domain_{\bm{\lambda}}$ & Set of placements under which system is stable under arrivals $\bm{\lambda}$\\
$\tilde{\domain}$ & Convex hull of constraints of \textsc{MaxCG}\\
\hline
\multicolumn{2}{c}{\textbf{Conventions}}\\
\hline
$\supp(\cdot)$ & Support of a vector\\
$\conv(\cdot)$ & Convex hull of a set\\
$[\mathbf{x}]_{+i}$ & Vector equal to $\mathbf{x}$ with $i$-th coordinate set to 1 \\
$[\mathbf{x}]_{-i}$ & Vector equal to $\mathbf{x}$ with $i$-th coordinate set to 0 \\
$\mathbf{0}$ & Vector of zeros\\
\hline
\end{tabular}
\vspace*{-2em}
\end{scriptsize}
\end{table}

	
	\section{Submodularity and the Greedy Algorithm}\label{sec:subm}
Problem \textsc{MinCost} is NP-hard; this is true even when cost functions $c_e$ are linear, and the objective is to minimize the sum of the loads per edge \cite{ioannidis2016adaptive,femtocaching}. In what follows, we outline our methodology for solving this problem; it relies on the fact that the objective of \textsc{MinCost} is a \emph{supermodular set function}; our first main contribution is to show that this property is a direct consequence of Assumption~\ref{as:conv}. 


\noindent\textbf{Cost Supermodularity and Caching Gain.}
 First, observe that the cost function $C$ in \textsc{MinCost} can be naturally  expressed as a set function. Indeed, for $S\subset V\times \catalog$, let $\mathbf{x}_S\in \{0,1\}^{|V| |\catalog|} $ be the binary vector whose support  is $S$ (i.e., its non-zero elements are indexed by $S$). As there is a 1-1 correspondence between a binary vector $\mathbf{x}$ and its support $\supp(\mathbf{x})$,  we can interpret $C:\{0,1\}^{|V| |\catalog|}\to \reals_+$ as  set function $C:V\times \catalog:\to\reals_+$ via
$C(S)\triangleq C(\mathbf{x}_S).$ Then, the following theorem holds:
\begin{theorem}\label{thm:coststruct}
Under Assumption~\ref{as:conv},  
$C(S)\triangleq C(\mathbf{x}_S)$
 is non-increasing and supermodular  over  $\{\supp(\mathbf{x}): \mathbf{x}\in \domain_{\bm{\lambda}}\}$.
\end{theorem}
\techrep{A detailed proof of Theorem~\ref{thm:coststruct} can be found in Appendix~\ref{app:coststruct}.}
{\begin{proof}[Proof (Sketch)]
It is easy to verify that $\rho_{e}(S), \forall e\in E$, is supermodular and non-increasing in $S$. By Assumption \ref{as:conv} $C_e$ is a non-decreasing function, so $C_e(S)\triangleq C_{e}(\rho_{e}(S))$ is non-increasing. The theorem then follows from Assumption \ref{as:conv} and the fact that if $f:\reals\to\reals$ is  convex and non-decreasing and $g: \mathcal{X}\to \reals$ is  non-increasing supermodular then $h(\mathbf{x})\triangleq f(g(\mathbf{x}))$ is also supermodular.\end{proof}}
 In light of the observations in Sec.~\ref{sec:mincost} regarding Assumption~\ref{as:conv},  Thm.~\ref{thm:coststruct} implies that supermodularity arises \emph{for a broad array of natural cost objectives}, including expected delay and system size; it also applies under the \emph{full generality of Kelly networks}, including FIFO, LIFO, and round robin service disciplines.
Armed with this theorem, we turn our attention to converting \textsc{MinCost} to a \emph{submodular maximization} problem. In doing so, we face the problem that domain $\domain_{\bm{\lambda}}$, determined not only by storage capacity constraints, but also by stability, may be difficult to characterize. Nevertheless, we show that a problem that is amenable to approximation can be constructed, provided that a placement $\mathbf{x}_0 \in \domain_{\bm{\lambda}}$ is known.

	
In particular, suppose that we have access to a single $\mathbf{x}_0\in \domain_{\bm{\lambda}}$. We define the \emph{caching gain} $F: \domain_{\bm{\lambda}} \to \reals_+$ as $F(\mathbf{x}) = C(\mathbf{x}_0)- C(\mathbf{x}). $
Note that, for $\mathbf{x}\geq \mathbf{x}_0$, $F(\mathbf{x})$ is the relative decrease in the cost compared to the cost under $\mathbf{x}_0$. We consider the following optimization problem:\\
	\begin{subequations}\label{opt:maxcg}
{\hspace*{\stretch{1}} \textsc{MaxCG} \hspace*{\stretch{1}} }\vspace{-0.8em}
		\begin{align}
			\text{Maximize:}&\quad F(\mathbf{x}) = C(\mathbf{x}_0) - C(\mathbf{x})\label{eq:fobj}\\   
			\text{subj.~to:}&\quad \mathbf{x} \in\domain, \mathbf{x}\geq \mathbf{x}_0 \label{eq:fcon}
	        \end{align} 
	\end{subequations}
Observe that, if $\mathbf{0}\in \domain_{\bm{\lambda}}$, then $\domain_{\bm{\lambda}} =\domain$; in this case, taking $\mathbf{x}_0=\mathbf{0}$ ensures that problems \textsc{MinCost} and \text{MaxCG} are equivalent. If $\mathbf{x}_0\neq \textbf{0}$, the above formulation attempts to maximize the gain restricted to placements $\mathbf{x}\in \domain$ that dominate $\mathbf{x}_0$: such placements necessarily satisfy $\mathbf{x}\in \domain_{\bm{\lambda}}$. 
Thm.~\ref{thm:coststruct} has the following immediate implication: 
\begin{cor}The caching gain $F(S)\triangleq F(\mathbf{x}_S)$ is non-decreasing and submodular  over  $\{\supp(\mathbf{x}): \mathbf{x}\in \domain_{\bm{\lambda}}\}$.
\end{cor}	

\noindent\textbf{Greedy Algorithm.}\label{sect:greedy}
Constraints \eqref{eq:fcon} define a (partition) matroid \cite{calinescu2011,femtocaching}. This, along with the submodularity and monotonicity of $F$ imply that we can produce a solution  within $\frac{1}{2}$-approximation from the optimal via the \emph{greedy} algorithm \cite{greedy}. The algorithm, summarized in Alg.~\ref{alg:greedy}, iteratively allocates items to caches that yield the largest marginal gain.
The solution produced by Algorithm \ref{alg:greedy}  is guaranteed  to be  within a $\frac{1}{2}$-approximation ratio of the optimal solution of \textsc{MaxCG} \cite{greedy2}.  The approximation guarantee of $\frac{1}{2}$ is tight: 
\begin{lemma}\label{lem:tight} For any $\varepsilon>0$, there exists a cache network  the greedy algorithm solution is within $\frac{1}{2}+\varepsilon$ from the optimal, when the objective is the sum of expected delays per edge.
\end{lemma}
\techrep{The proof of Lemma~\ref{lem:tight} can be found in Appendix \ref{app:proofoflem:tight}.}{} The instance under which the bound is tight is given in Fig.~\ref{fig:halfistight}\techrep{}{ (see also \cite{techrep})}.
As we discuss in Sec.~\ref{sec:experiments}, the greedy algorithm performs well in practice for some topologies; however, Lemma~\ref{lem:tight} motivates us to seek alternative algorithms, that attain  improved approximation guarantees. 

\begin{figure}[!t]
\centering 
    \includegraphics[width=0.8\columnwidth]{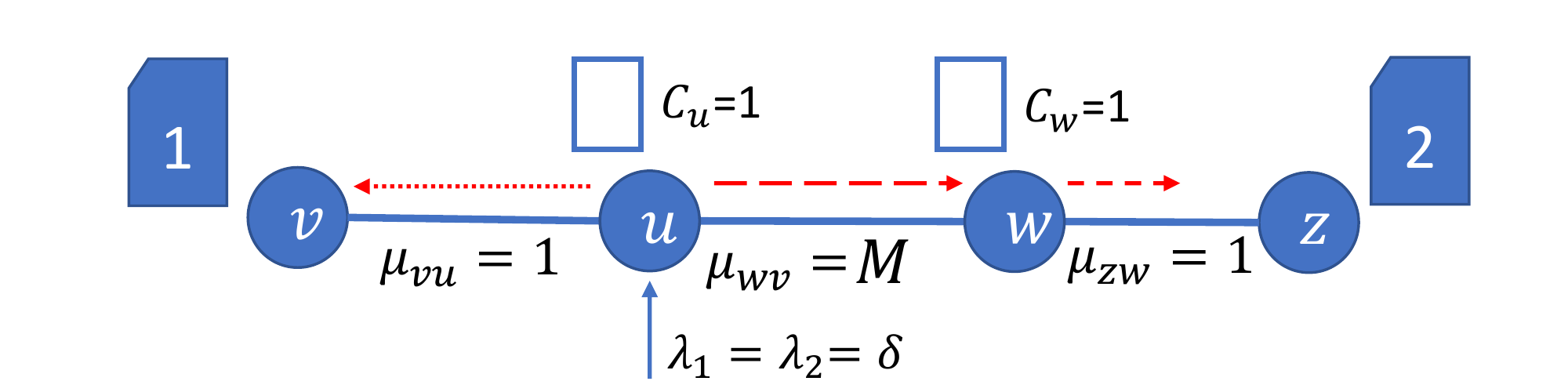}
    \caption{A path graph, illustrating that the $1/2$-approximation ratio of greedy is tight. Greedy caches item $2$ in node $u$, while the optimal decision is to cache item $1$ in $u$ and item $2$ in node $w$. For $M$ large enough, the approximation ratio can be made arbitrarily close to 1/2.  In our experiments in Sec.~\ref{sec:experiments}, we set $\delta=0.5$ and $M=200$.}\label{fig:halfistight}
\end{figure}

  \algsetup{linenosize=\scriptsize} 
\begin{algorithm}[!t]
		\caption{Greedy}\label{alg:greedy}
\begin{scriptsize}
		\begin{algorithmic}[1]
			\REQUIRE $F:\domain \rightarrow \reals_+,\mathbf{x}_0$
			\STATE   $\mathbf{x}\leftarrow \mathbf{x}_0$
			\WHILE{ $A(\mathbf{x}):= \{(v,i)\in V\times \catalog: \mathbf{x}+\mathbf{e}_{vi} \in \domain \}$  is not empty }
			\STATE  $(v^*,i^*) \leftarrow \mathop{\arg \max}_{(v,i) \in A(\mathbf{x})} \left( F(\mathbf{x}+\mathbf{e}_{vi}) -F(\mathbf{x}) \right)$ \label{line: max} 
			\STATE $\mathbf{x} \leftarrow \mathbf{x}+\mathbf{e}_{v^*i^*}$
			\ENDWHILE
			\RETURN $\mathbf{x}\;$
		\end{algorithmic}
\end{scriptsize}
	\end{algorithm}

	
%

\section{Continuous-Greedy Algorithm}\label{sec:contgreed}
The \emph{continuous-greedy} algorithm by  Calinescu et al.~\cite{calinescu2011} attains a tighter guarantee than the greedy algorithm, raising the approximation ratio from $0.5$ to $1-1/e\approx 0.63$. The algorithm maximizes the so-called  \emph{multilinear extension} of objective $F$, thereby obtaining a fractional solution $Y$ in the convex hull of the constraint space. The resulting solution is then rounded to produce an integral solution.

	\begin{algorithm}[!t]
		\caption{Continuous-Greedy}\label{alg:cg}
\begin{scriptsize}
		\begin{algorithmic}[1]
			\REQUIRE $G:\tilde{\domain}\to\reals_+$, $\mathbf{x}_0$, stepsize $0<\gamma\leq 1$
			\STATE  $t \leftarrow 0, k\leftarrow 0\; \mathbf{y}_0\leftarrow \mathbf{x}_0$
			\WHILE{$t <1$}
			\STATE  $\mathbf{m}_k \leftarrow \mathop{\arg\max}_{\mathbf{m}\in \tilde{\domain}} \langle \mathbf{m},\nabla G(\mathbf{y}_k)\rangle\;$ 
			\STATE $\gamma_k\leftarrow \min\{\gamma,1-t\}\;$
			\STATE $\mathbf{y}_{k+1}\leftarrow \mathbf{y}_k+\gamma_k\mathbf{m}_k, t\leftarrow t+\gamma_k, k\leftarrow k+1\;$
			\ENDWHILE
			\RETURN $\mathbf{y}_k\;$
		\end{algorithmic}
\end{scriptsize}
	\end{algorithm}

\subsection{Algorithm Overview} Formally, the multilinear extension of the caching gain $F$  is defined as follows. Define the convex hull of the set defined by the constraints \eqref{eq:fcon} in \textsc{MaxCG} as:
 \begin{align}\!\!\!\tilde{\domain} = \conv ( \{\mathbf{x} 
 : \mathbf{x}\in \domain, \mathbf{x}\geq \mathbf{x}_0  \} )  \subseteq[0,1]^{|V||\mathcal{C}|}\label{eq:hull}\end{align}
 Intuitively, $\mathbf{y}\in \tilde{\domain}$ is a \emph{fractional} vector in $\reals^{|V||\domain|}$ satisfying the capacity constraints, and the bound $\mathbf{y}\geq \mathbf{x}_0$.

 Given a $\mathbf{y}\in \tilde{\domain}$, consider a random vector $\mathbf{x}$ in $\{0,1\}^{|V||\mathcal{C}|}$ generated as follows: for all $v\in V$ and $i \in \catalog$, the coordinates $x_{vi}\in\{0,1\}$  are independent Bernoulli  variables such that  $\prob(x_{vi}=1)=y_{vi}$. The multilinear extension $G:\tilde{\domain}\to \reals_+$ of $F:\domain_{\bm{\lambda}}\to\reals_+$ is defined via  following expectation $G(\mathbf{y}) = \mathbb{E}_{\mathbf{y}}[F(\mathbf{x})]$, parameterized by $\mathbf{y}\in\tilde{\domain}$, i.e.,
	\begin{equation}
	\label{eq:multilinearEx}
	G(\mathbf{y})  = \!\!\!\! \sum\limits_{\mathbf{x}\in \{0,1\}^{|V||\mathcal{C}|}}\!\!\!\!\!\!F(\mathbf{x}) \times\!\!\!\! \prod_{(v,i)\in V\times \catalog}\!\!\!\!\!y_{vi}^{x_{vi}}(1-y_{vi})^{1-x_{vi}},\!\!
	\end{equation}

The continuous-greedy algorithm, summarized in Alg.~\ref{alg:cg}, proceeds by first producing a fractional vector $\mathbf{y}\in \tilde{\domain}$. Starting from $\mathbf{y}_0=\mathbf{x}_0$, the algorithm iterates over: 
\begin{subequations}\label{eq:its}
\begin{align}
	\mathbf{m}_k &\in \textstyle\mathop{\arg\max}_{\mathbf{m}\in\tilde{\domain}} \langle \mathbf{m}, \nabla G(\mathbf{y}_k)\rangle, \label{eq:maxlin}\\
        \mathbf{y}_{k+1} &= \mathbf{y}_k + \gamma_k \mathbf{m}_k,
\end{align} 
\end{subequations}
for an appropriately selected step size $\gamma_k\in [0,1]$. Intuitively, this yields an approximate solution to the non-convex problem:
\begin{subequations}\label{eq:nonconv}
\vspace*{-5mm}
 \begin{align}
\text{Maximize:} &\quad G(\mathbf{y})\\
\text{subj.~to:} & \quad \mathbf{y}\in\tilde{\domain}.
\end{align}
\end{subequations}
Even though \eqref{eq:nonconv} is not  convex, the output of Alg.~\ref{alg:cg} is  within a $1-1/e$ factor from the optimal solution $\mathbf{y}^*\in \tilde{\domain}$ of \eqref{eq:nonconv}. This fractional solution can  be rounded  to produce a solution to \textsc{MaxCG} with the same approximation guarantee using  either the pipage rounding \cite{ageev2004pipage} or the swap rounding \cite{calinescu2011,chekuri2010dependent} schemes: \techrep{we review both  in Appendix~\ref{sec:rounding}.}{ we refer the reader to  \cite{techrep} for details.}

\noindent\textbf{A Sampling-Based Estimator.}\label{sec:sample}
Function $G$, given by \eqref{eq:multilinearEx}, involves a summation over $2^{|V||\catalog|}$ terms, and cannot be easily computed in polynomial time. 
Typically, a sampling-based estimator (see, e.g., \cite{calinescu2011})  is used instead.  Function $G$ is linear when restricted to each coordinate $y_{vi}$, for some $v\in V$, $i\in \catalog$ (i.e., when all  inputs except $y_{vi}$ are fixed). As a result, the partial derivative of $G$ w.r.t. $y_{vi}$ can be written as:
	\begin{align}\label{eq:deriv_G}
\textstyle	\frac{\partial G(\mathbf{y})}{\partial y_{vi}}& = \mathbb{E}_{\mathbf{y}}[F(\mathbf{x})|x_{vi}=1]- \mathbb{E}_{\mathbf{y}}[F(\mathbf{x})|x_{vi}=0]\geq 0,
	\end{align} 
	where the last inequality is due to monotonicity of $F$. One can thus estimate the gradient by 
(a) producing $T$ random samples $\mathbf{x}^{(\ell)}$, $\ell=1,\ldots,T$ of the random vector $\mathbf{x}$, consisting of independent Bernoulli coordinates,  and (b) computing, for each pair $(v,i)\in V\times \catalog$, the average
\begin{align}\textstyle
\widehat{\frac{\partial G(\mathbf{y})}{\partial y_{vi}}} = 
\frac{1}{T} \sum_{\ell=1}^T \left( F([\mathbf{x}^{\ell}]_{+(v,i)} )-F([\mathbf{x}^{\ell}]_{-(v,i)})\right),\label{eq:samplestimator}
\end{align}
where $[\mathbf{x}]_{+(v,i)}$,$[\mathbf{x}]_{-(v,i)}$ are equal to vector $\mathbf{x}$ with the $(v,i)$-th coordinate set to 1 and 0, respectively. Using this estimate, Alg.~\ref{alg:cg} \techrep{}{is poly-time and}
 attains an approximation ratio arbitrarily close to $1-1/e$ for appropriately chosen $T$\techrep{.}{ (see \cite{techrep} for a formal guarantee and characterization of Alg.~\ref{alg:cg}'s complexity).}
\techrep{
In particular, the following theorem holds:
\begin{theorem}\label{thrm:sampling}[Calinescu et al.~\cite{calinescu2011}] Consider Alg.~\ref{alg:cg}, with $\nabla G(\mathbf{y}_k)$ replaced by the sampling-based estimate $\widehat{\nabla G(\mathbf{y}^k)}$, given by \eqref{eq:samplestimator}. Set $T = \frac{10}{\delta^2}(1+\ln (|\catalog||{V}|))$, and $\gamma=\delta$, where $\delta = \frac{1}{40|\catalog||{V}|\cdot(\sum_{v\in{V}}c_v)^2}.$ Then, the algorithm terminates after $K=1/\gamma=1/\delta$ steps and, with high probability, 
$$G(\mathbf{y}^K) \geq (1-(1-\delta)^{1/\delta}) G(\mathbf{y}^*) \geq (1-1/e)G(\mathbf{y}^*),$$
where $\mathbf{y}^*$ is an optimal solution to \eqref{eq:nonconv}.
\end{theorem}
The proof of the theorem can be found  in Appendix A of Calinescu et al.~\cite{calinescu2011} for general submodular functions over arbitrary matroid constraints; we state Thm.~\ref{thrm:sampling} here with constants $T$ and $\gamma$ set specifically for our objective $G$ and our set of constraints $\tilde{D}$.}{}

\techrep{Under this parametrization of $T$ and $\gamma$, Alg.~\ref{alg:cg} runs in polynomial time.  More specifically, note that $1/\delta=  O(|\catalog||{V}|\cdot(\sum_{v\in{V}}c_v)^2)$ is polynomial in the input size. Moreover, the algorithm runs for $K =1/\delta$ iterations in total. Each iteration requires $T = O(\frac{1}{\delta^2}(1+\ln (|\catalog||V|)$ samples, each involving a polynomial computation (as $F$ can be evaluated in polynomial time). Finally, LP \eqref{eq:maxlin} can be solved in polynomial time in the number of constraints and variables, which are $O(|V||\catalog|)$.}{}

%
 
\subsection{A Novel Estimator via Taylor Expansion}\label{sec:taylor}

The classic approach to estimate the gradient via sampling has certain drawbacks.  The number of samples $T$ required to attain the $1-1/e$ ratio is quadratic in $|V||\catalog|$. In practice, even for networks and catalogs of moderate size (say, $|V|=|\catalog|=100$), the number of samples becomes prohibitive (of the order of $10^8$). Producing an estimate for $\nabla G$ via a closed form computation that eschews sampling thus has significant computational advantages. In this section, we show that the multilinear relaxation of the caching gain $F$ admits  such a closed-form characterization.

  We say that a polynomial $f:\reals^d\to\reals$ is in  \emph{Weighted Disjunctive Normal Form} (\emph{W-DNF}) if it can be written as
\begin{align}
	f(\mathbf{x}) = \textstyle\sum_{s\in \mathcal{S}}\beta_s \cdot \prod_{j\in \mathcal{I}(s)}(1-x_j),\label{wdnf}
\end{align}
for some index set $\mathcal{S}$, positive coefficients $\beta_s>0$, and index sets $I(s)\subseteq \{1,\ldots, d\}$. Intuitively,  treating binary variables $x_j\in \{0,1\}$ as boolean values,  each W-DNF polynomial can be seen as a weighted  sum (disjunction) among products (conjunctions) of negative literals.  
These polynomials  arise naturally in the context of our problem; in particular:
\begin{lemma}\label{wdnfmoments}
For all $k\geq 1$, $\mathbf{x}\in \mathcal{D}$, and $e\in E$,  $\rho^k_e(\mathbf{x},\bm{\lambda})$ is a W-DNF polynomial whose coefficients depend on $\bm{\lambda}$.
\end{lemma}
\begin{proof}[Proof (Sketch)]
The lemma holds for $k=1$ by \eqref{lamr} and \eqref{eq:rho_X}. The lemma follows by induction, as W-DNF polynomials over binary  $\mathbf{x}\in \mathcal{D}$ are closed under multiplication; this is because $(1-x)^\ell=(1-x)$ for all $\ell\geq 1$ when $x\in\{0,1\}$.
\end{proof}
  Hence, \emph{all load powers are W-DNF polynomials}\techrep{; a detailed proof be found in Appendix~\ref{app:wdnfmomentslemma}}{}.
Expectations of  W-DNF polynomials have a remarkable property:
\begin{lemma}\label{lemma:compute} Let  $f:\domain_{\bm{\lambda}}\to\reals$ be a W-DNF polynomial, and let $\mathbf{x}\in \domain$ be a random vector of independent Bernoulli coordinates parameterized by $\mathbf{y}\in \tilde{\domain}$. Then $\expect_{\mathbf{y}}[f(\mathbf{x})]=f(\mathbf{y})$, where $f(\mathbf{y})$ is the evaluation of the W-DNF polynomial representing $f$ over the real vector $\mathbf{y}$.  
\end{lemma} 
\techrep{
\begin{proof}
As $f$ is W-DNF, it can be written as 
\begin{align*}f(\mathbf{x})=\sum_{s\in \mathcal{S}}\beta_s \prod_{t\in \mathcal{I}(s)}(1-x_t)
\end{align*}
for appropriate $\mathcal{S}$, and appropriate $\beta_s,\mathcal{I}(s)$, where $s\in\mathcal{S}$. Hence,
\begin{align*}\expect_{\mathbf{y}}[f(\mathbf{x})]&=\sum_{s\in \mathcal{S}}\beta_s \expect_{\mathbf{y}}\left[\prod_{t\in \mathcal{I}(s)}(1-x_t)\right]\\
& = \sum_{s\in \mathcal{S}}\beta_s \prod_{t\in \mathcal{I}(s)}(1-\expect_{\mathbf{y}}[x_t]), \text{ by independence}\\
& = \sum_{s\in \mathcal{S}}\beta_s \prod_{t\in \mathcal{I}(s)}(1-y_t).\qedhere
\end{align*}
\end{proof}
}{
\begin{proof}
By \eqref{wdnf} and independence, we have
$\expect_{\mathbf{y}}[f(\mathbf{x})]=\sum_{s\in \mathcal{S}}\beta_s \expect_{\mathbf{y}}\left[\prod_{t\in \mathcal{I}(s)}(1-x_t)\right] 
 = \sum_{s\in \mathcal{S}}\beta_s \prod_{t\in \mathcal{I}(s)}(1-\expect_{\mathbf{y}}[x_t]) 
 = \sum_{s\in \mathcal{S}}\beta_s \prod_{t\in \mathcal{I}(s)}(1-y_t).$
\end{proof}}

Lemma~\ref{lemma:compute} states that, to compute the expectation of a W-DNF polynomial $f$ over i.i.d.~Bernoulli variables with expectations $\mathbf{y}$, it suffices to \emph{evaluate $f$ over input $\mathbf{y}$}. Expectations computed this way therefore do not require sampling. 

We  leverage this property to approximate $\nabla G(\bm{y})$ by taking the Taylor expansion of the cost functions $C_e$ at each edge $e\in E$. This allows us to write $C_e$ as a power series w.r.t.~$\rho_e^k$, $k\geq 1$; from Lemmas~\ref{wdnfmoments} and~\ref{lemma:compute}, we can compute the expectation of this series in a closed form. In particular, by expanding the series and rearranging terms it is easy to show the following\techrep{ lemma, which is proved in Appendix~\ref{proofoftaylorlemma}}{}:
\begin{lemma}\label{taylorlemma}
 Consider a cost function $C_e:[0,1)\to\reals_+$ which satisfies Assumption~\ref{as:conv} and for which the Taylor expansion exists at some $\rho^*\in[0,1)$. Then, for  $\mathbf{x} \in \domain$ a random Bernoulli vector parameterized by $\mathbf{y}\in \tilde{\domain}$,
\begin{align}
\!\!\!\frac{\partial G(\mathbf{y})}{\partial y_{vi}} \!\!\approx\!\!
\sum_{e\in E}\! \sum_{k=1}^{L}\!\alpha^{(k)}_e \!\!\left[\rho^k_e\left([\mathbf{y}]_{-(v,i)},\!\bm{\lambda}\right) \!-\!\rho^k_e\left([\mathbf{y}]_{+(v,i)},\!\bm{\lambda}\right)\right]\!\!\!
\label{eq:taylorapprox2}
\end{align}
where, 	$\textstyle\alpha^{(k)}_e  = \sum_{i=k}^{L} \frac{(-1)^{i-k}\binom{i}{k}}{i!}C^{(i)}_{e}(\rho^*)(\rho^*)^{i-k}$  for $k=0,1,\cdots,L,$ and the error of the approximation is:
$\textstyle\frac{1}{(L+1)!}\sum_{e\in E}C^{(L+1)}_{e}(\rho') \Big[ \expect_{[\mathbf{y}]_{-{v,i}}}[(\rho_e(\mathbf{x},\bm{\lambda})-\rho^*)^{L+1}] 
\textstyle-  \expect_{[\mathbf{y}]_{+{v,i}}}[(\rho_e(\mathbf{x},\bm{\lambda})-\rho^*)^{L+1}]\Big]. $

\end{lemma}
 Estimator \eqref{eq:taylorapprox2} is \emph{deterministic}: no random sampling is required. Moreover, Taylor's theorem allows us to characterize the error  (i.e., the \emph{bias}) of this estimate. We use this to characterize the final fractional solution $\mathbf{y}$ produced by Alg.~\ref{alg:cg}:
  	\begin{theorem}\label{theorem:converge}
	     Assume that all $C_e$, $e\in E$, satisfy Assumption~\ref{as:conv}, are $L+1$-differentiable, and that all their $L+1$ derivatives are bounded by $W\geq 0$. Then, consider Alg.~\ref{alg:cg}, in which $\nabla G(\mathbf{y}_k)$ is estimated via the Taylor estimator \eqref{eq:taylorapprox2}, where each edge cost function is approximated at $\rho_e^* = \expect_{\mathbf{y}_k}[\rho_e(\mathbf{x},\bm{\lambda})]=\rho_e(\mathbf{y}_k,\bm{\lambda}).$
Then,  
		\begin{equation}\label{eq:optimal}
		\textstyle G(\mathbf{y}_{K}) \geq (1-\frac{1}{e})G(\mathbf{y}^*) - 2D B- \frac{P}{2K},
		\end{equation}
		where  $K=\frac{1}{\gamma}$ is the number of iterations, $\mathbf{y}^*$ is an optimal solution to \eqref{eq:nonconv},  
$D = \max_{\mathbf{y}\in \tilde{\domain}} \|\mathbf{y}\|_2\leq |V|\cdot\max\limits_{v\in{V}}c_v,$ is the diameter of  $\tilde{\domain}$, $B\leq \frac{W|E|}{(L+1)!}$ is the bias of estimator \eqref{eq:taylorapprox2}, and $P=2C(\mathbf{x}_0),$ is a Lipschitz constant of $\nabla G$.
	\end{theorem}
	\techrep{The proof can be found in Appendix~\ref{app:3}.}{} The theorem immediately implies that we can replace \eqref{eq:taylorapprox2} as an estimator in Alg.~\ref{alg:cg}, and attain an approximation arbitrarily close to $1-1/e$.

	\noindent\textbf{Estimation via Power Series.} For arbitrary $L+1$-differentiable cost functions $C_e$, the estimator \eqref{eq:taylorapprox2} can be leveraged by replacing $C_e$ with its Taylor expansion. In the case of queue-dependent cost functions, as described in Example 4 of Section~\ref{sec:mincost}, the power-series \eqref{eq:power} can be used instead. For example, the expected queue size (Example 1, Sec.~\ref{sec:mincost}), is given by $C_e(\rho_e) = \frac{\rho_e}{1-\rho_e} =\sum_{k=1}^\infty \rho_e^k.$
In contrast to the Taylor expansion, this power series does not depend on a point $\rho^*_e$ around which the function $C_e$ is approximated.

	\section{Beyond M/M/1 queues}\label{sec:beyond}
    	As discussed in Section~\ref{sec:model}, the classes of M/M/1 queues for which the supermodularity of the cost functions arises is quite broad, and includes FIFO, LIFO, and processor sharing queues. In this section, we discuss how our results extend to even broader families of queuing networks.  
		\techrep{Chapter 3 of Kelly~\cite{kelly} provides a general framework for a set of queues for which service times are exponentially distributed; for completeness, we also summarize this  in Appendix~\ref{app:kelly}. 
		A large class of networks can be modeled by this framework, including networks of M/M/$k$ queues; all such networks maintain the property that steady-state distributions have a product form.  
This allows us to}{We first} extend our results to M/M/$k$ queues for two cost functions $C_e$: 
		\begin{lemma}\label{lemma:mmk}
			For a network of M/M/k queues, both the  queuing probability\footnote{This is 
			given by the so-called Erlang C formula \cite{datanetworks}.} and the expected queue size are non-increasing and supermodular over sets $\{\supp(\mathbf{x}): \mathbf{x}\in \domain_{\bm{\lambda}}\}$.
		\end{lemma}
We note that, as an immediate consequence of Lemma~\ref{lemma:mmk} and Little's theorem, both the sum of the expected delays per queue, but also the expected delay of an arriving packet, are also  supermodular and non-decreasing.
		
		Product-form  steady-state distributions arise also in settings where service times are not exponentially distributed. A large class  of quasi-reversible queues, named {\em symmetric queues} exhibit this property (c.f.~Section 3.3 of \cite{kelly} and Chapter 10 of \cite{nelson}). \techrep{For completeness, we again summarize symmetric queues in Appendix \ref{app:symmetric}.}{}		
In the following lemma 
 we leverage the product form of symmetric queues to extend our results to M/D/1 symmetric queues \cite{datanetworks}.
		\begin{lemma}\label{lemma:symmetric}
			For a network of M/D/1 symmetric queues, the expected queue size is non-increasing and supermodular over sets $\{\supp(\mathbf{x}): \mathbf{x}\in \domain_{\bm{\lambda}}\}$.
		\end{lemma}
		Again, Lemma~\ref{lemma:symmetric} and Little's theorem imply that this property also extends to network delays.
		It is worth noting that conclusions similar to these in Lemmas \ref{lemma:mmk} and \ref{lemma:symmetric} are not possible for all general queues with product form distributions. In particular, also we prove the following negative result: 
		\begin{lemma} There exists a network of  M/M/1/k queues, containing a queue $e$, for which no strictly monotone function $C_e$ of the load $\rho_e$ at a queue $e$ is non-increasing and supermodular over sets $\{\supp(\mathbf{x}): \mathbf{x}\in \domain_{\bm{\lambda}}\}$. In particular, the expected  size of queue $e$ is neither monotone nor supermodular. 
			\label{lemma:example}
		\end{lemma}
	

        \section{Numerical Evaluation}\label{sec:experiments}
	\begin{table}[!t]
\caption{Graph Topologies and Experiment Parameters.}\label{tab:networks}
\begin{scriptsize}
\begin{tabular}{  p{13mm} p{3mm} p{3mm} p{3mm} p{2mm} p{2mm} p{2.5mm} cc}
Graph & $|V|$ & $|E|$ & $|\catalog|$ & $|\demand|$ &  $|Q|$ & $c_v$ & $F_{\text{PL}}\!(\mathbf{x}_{\text{RND}})$ & $F_{\text{UNI}}\!(\mathbf{x}_{\text{RND}})$\\
\hline
 \texttt{ER} &100 & 1042  & 300 & 1K & 4 & 3& 2.75& 2.98\\
 \texttt{ER-20Q} &100 & 1042  & 300 & 1K & 20 & 3& 3.1& 2.88\\
\texttt{HC}&128 & 896  & 300 & 1K & 4 & 3& 2.25 & 5.23\\
\texttt{HC-20Q}&128 & 896  & 300 & 1K & 20 & 3& 2.52 & 5.99\\
 \texttt{star} & 100 & 198 & 300 & 1K & 4 & 3 & 6.08 & 8.3\\
 \texttt{path} & 4 & 3 & 2 & 2 &1 &1& 1.2 &1.2\\
 \texttt{dtelekom} & 68 & 546 & 300 & 1K & 4 & 3 &2.57 &3.66\\
  \texttt{abilene} & 11 & 28 & 4 &2& 2& 1/2 & 4.39 &4.39\\
  \texttt{geant} &22 & 66& 10&100 & 4& 2 &19.68&17.22
\end{tabular}
\end{scriptsize}
\end{table}
\begin{figure}[!t]
\vspace{-4mm}
\centering 
    \includegraphics[width=0.8\columnwidth]{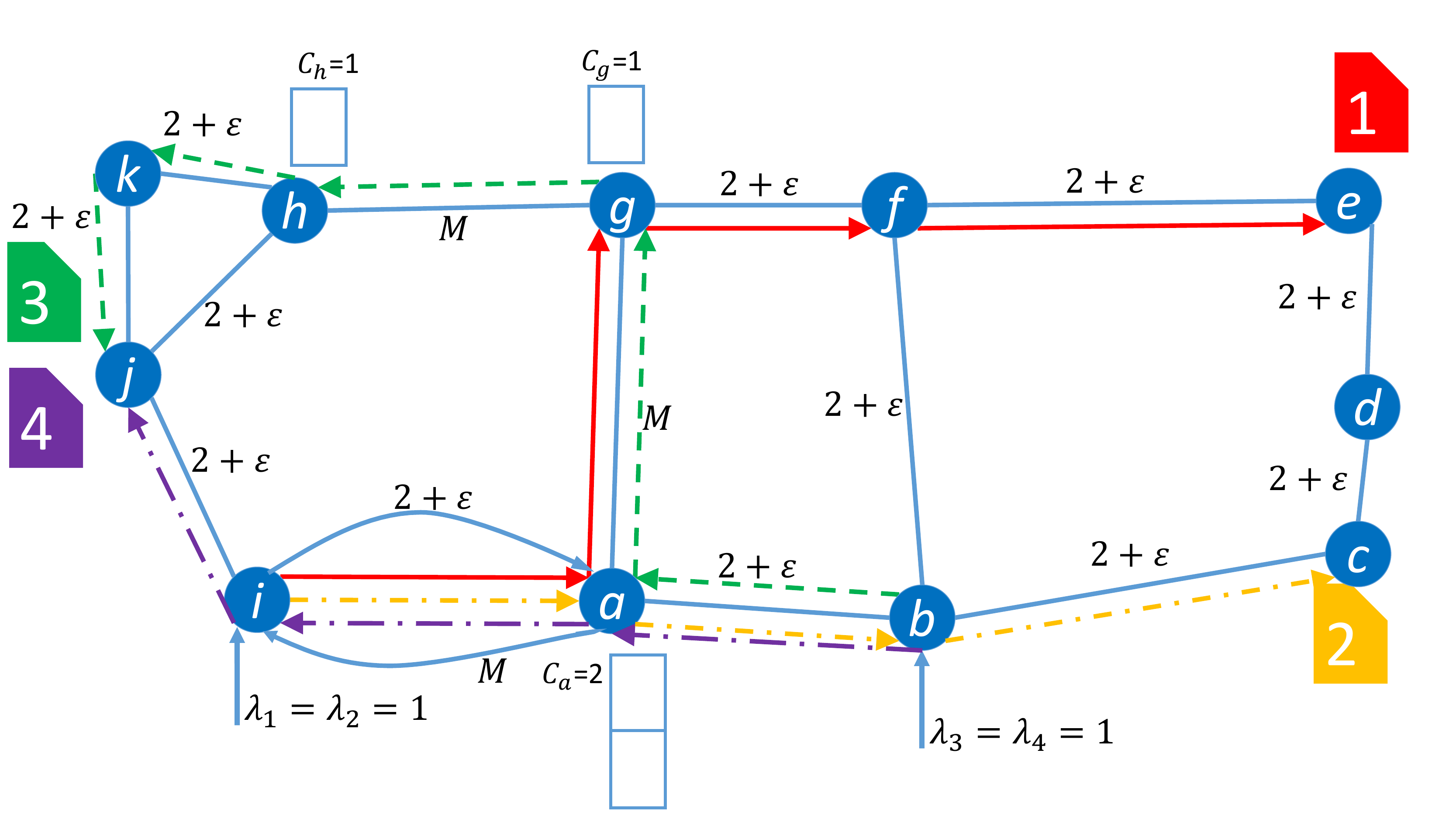}
    \caption{The \texttt{abilene} topology. We consider a catalog size of $|\catalog|=4$ and 4 requests $(|\demand|=4)$. Requests originate from $|Q|=2$ nodes, $b$ and $i$. Three edges have a high service rate $M\gg 1$, and the rest have a low service rate $2+ \varepsilon$. Only nodes $a, g,$ and $h$ can cache items, and have capacities  2, 1, and 1, respectively. We set $M=200$ and $\varepsilon=0.05$ in our experiments. Greedy is  0.5-approximate in this instance.}\label{fig:halfistight2}
\end{figure}
\noindent\textbf{Networks.} We execute Algorithms~\ref{alg:greedy}~and~\ref{alg:cg} over 9 network topologies, summarized in Table~\ref{tab:networks}. 
Graphs \texttt{ER} and \texttt{ER-20Q} are the same 100-node Erd\H{o}s-R\'enyi graph with parameter $p=0.1$. Graphs \texttt{HC} and \texttt{HC-20Q} are the same hypercube graph with 128 nodes, and graph \texttt{star} is a star graph with 100 nodes. The graph \texttt{path} is the topology shown in Fig.~\ref{fig:halfistight}. 
The last 3 topologies, namely, \texttt{dtelekom}, \texttt{geant}, and \texttt{abilene} represent the Deutsche Telekom, GEANT, and Abilene backbone networks, 
respectively. The latter is also shown in Fig.~\ref{fig:halfistight2}.

\begin{figure}[!t]
\vspace{-3mm}
	\centering
	\subfloat[Power-law demand]{
		\includegraphics[width=\columnwidth]{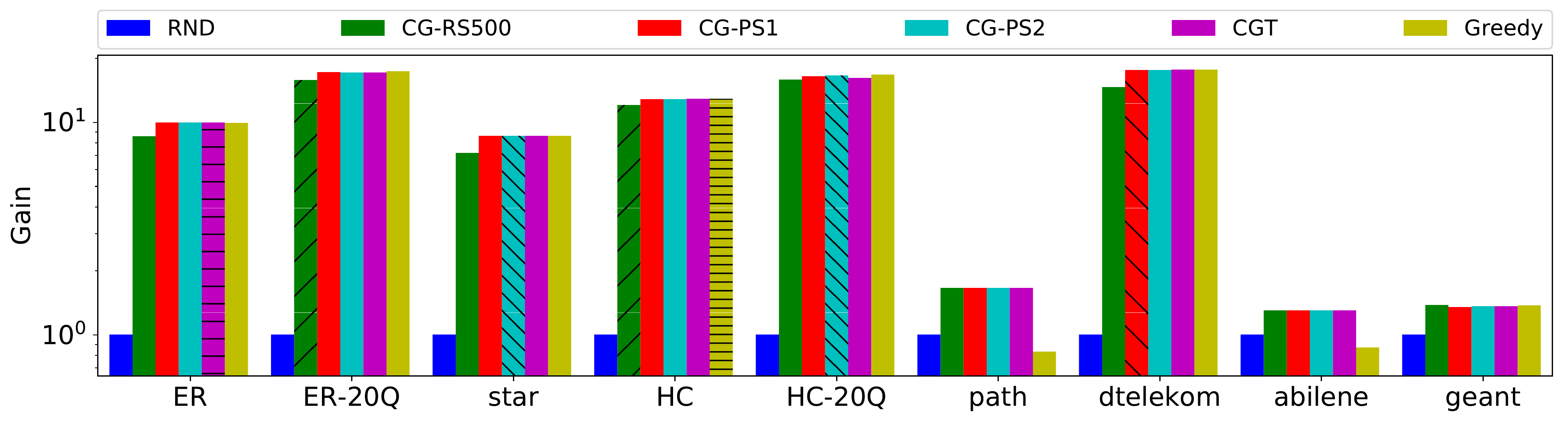}\label{cache_gain_plaw}
	}
	\vspace{-4mm}
	
	\subfloat[Uniform demand]{
		\includegraphics[width=\columnwidth]{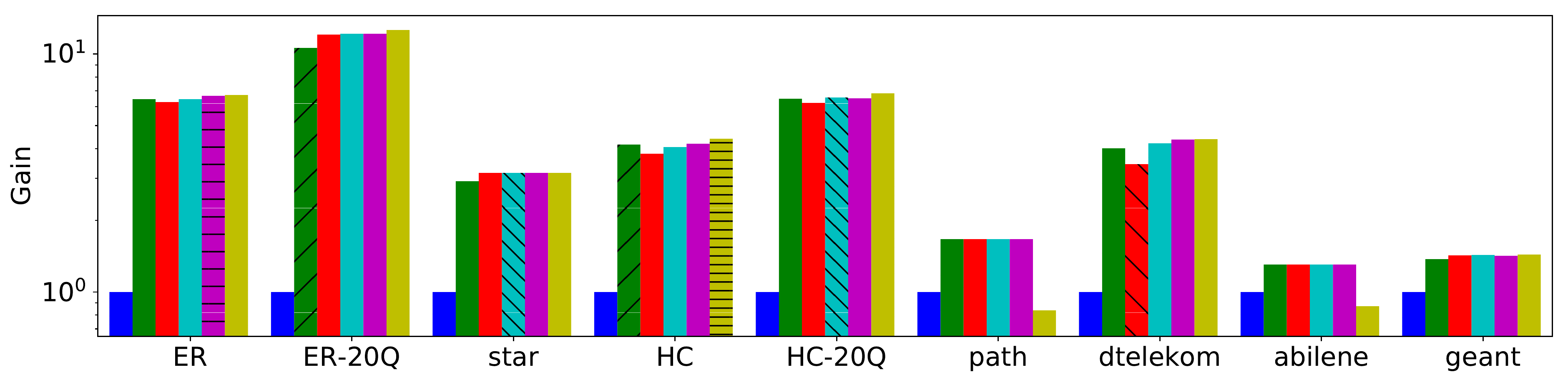}\label{cache_gain_uniform}
	}
\caption{Caching gains for different topologies and different arrival distributions, normalized by the gains corresponding to RND, reported in Table.~\ref{tab:networks}. Greedy performs comparatively well. However, it attains sub-optimal solutions for \texttt{path} and \texttt{abilene}; these solutions are worse than RND. CG-RS500 has a poor performance compared to other variations of the continuous-greedy algorithm.}\label{cache_gain_1}
\end{figure}
\begin{figure}[!t]\vspace{-4mm}
    \includegraphics[width=\columnwidth]{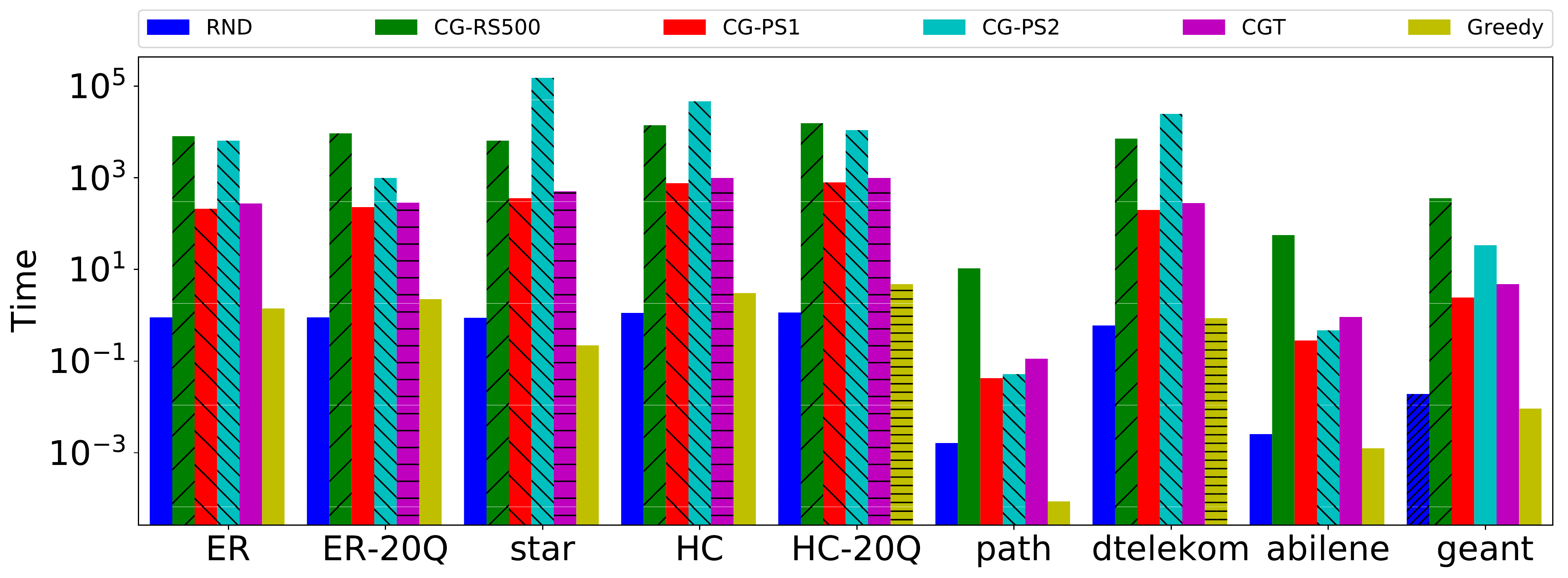}
    \caption{Running time for different topologies and power-law arrival distribution, in seconds. CG-RS500 is slower than power series estimation CG-PS1 and CGT, sometimes exceeding CG-PS2 as well.}\label{fig:time}
    \vspace{-2mm}
\end{figure}

\begin{figure}[!t]
\vspace{-4mm}
\centering

     \subfloat[\texttt{abilene}]{\includegraphics[width=0.5\columnwidth]{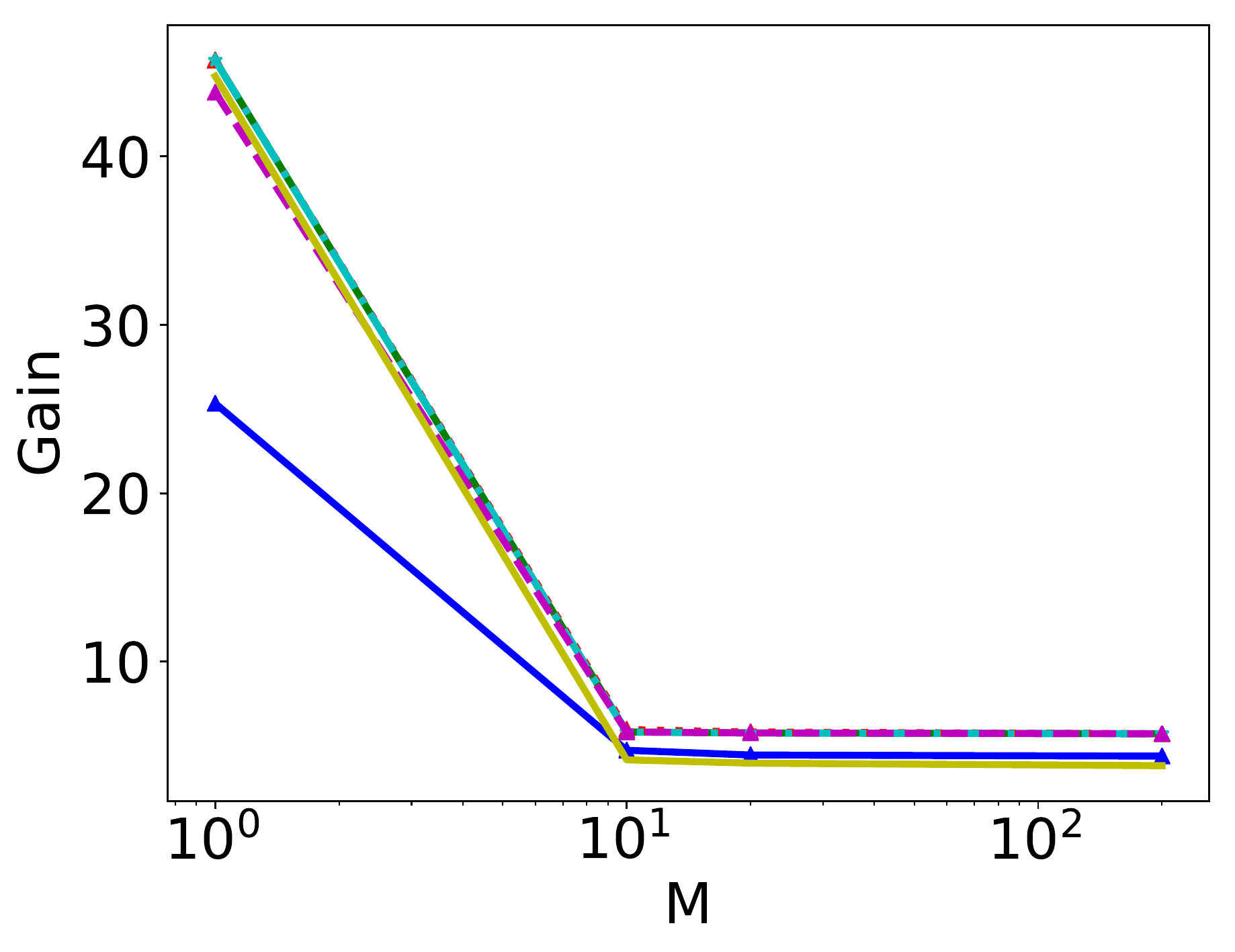}}
     \subfloat[\texttt{path}]{\includegraphics[width=0.5\columnwidth]{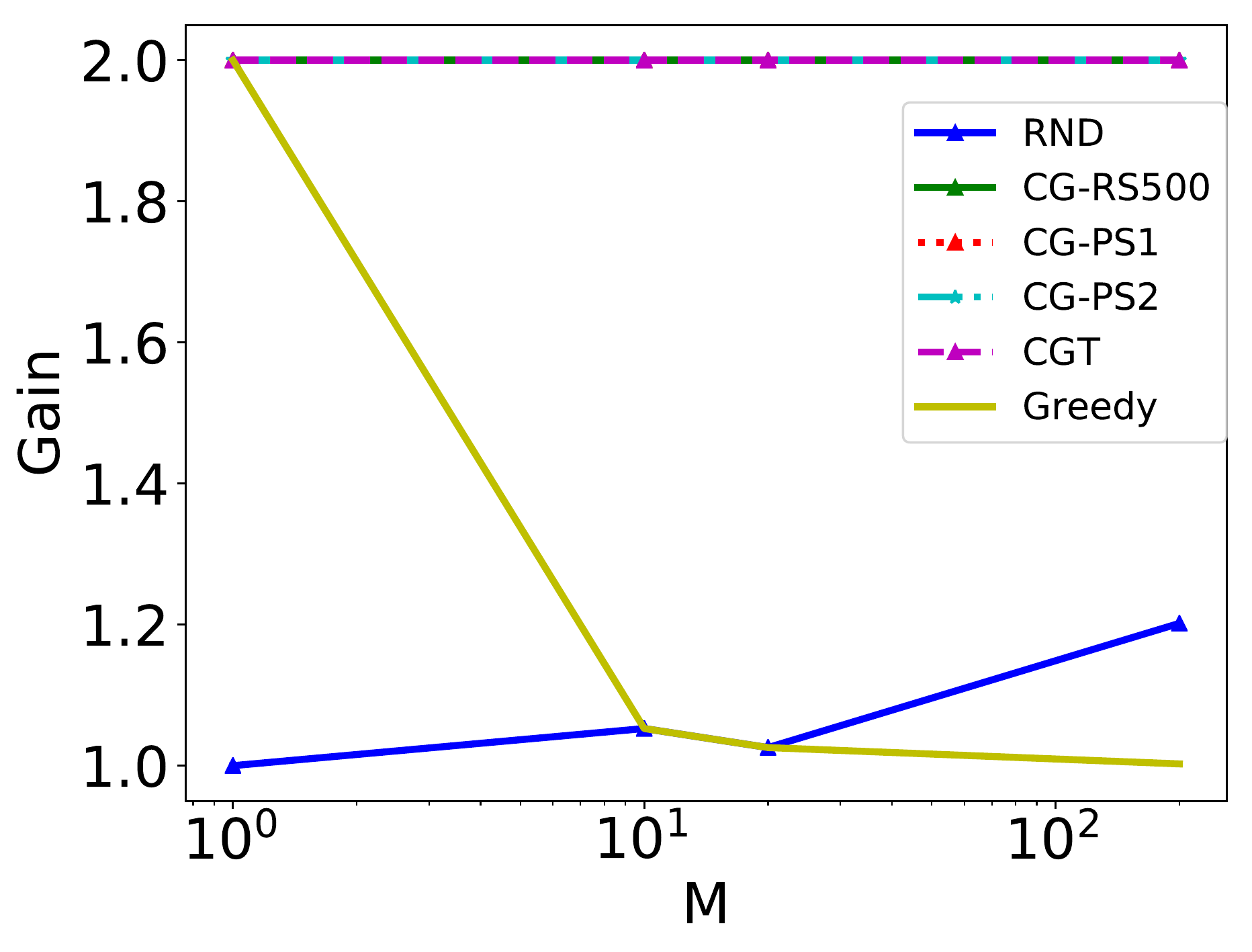}}
     
\caption{Caching gain vs. $M$. As the discrepancy between the service rate of low-bandwidth and high-bandwidth links increases, the performance of Greedy deteriorates.}\label{fig:M}
\end{figure}

\noindent\textbf{Experimental Setup.} 
For  \texttt{path} and \texttt{abilene}, we set demands, storage capacities, and service rates as illustrated in Figures~\ref{fig:halfistight} and \ref{fig:halfistight2}, respectively. Both of these settings induce an approximation ratio close to $1/2$ for greedy. 
For all remaining topologies,  we consider a catalog of size $|\catalog|$ objects; for each object, we select 1 node  uniformly at random (u.a.r.) from $V$ to serve as the designated server for this object. To induce traffic overlaps, we also select $|Q|$ nodes u.a.r.~that serve as sources for requests; all requests originate from these sources. 
All caches are set to the same storage capacity, i.e., $c_v=c$ for all $v\in V$. 
 We generate a set of $|\demand|$ possible  types of requests. For each request type $r\in \demand$, $\lambda^r=1$ request per second, and   path $p^r$ is generated by selecting a source among the $|Q|$ sources u.a.r., and routing towards the designated server of object $i^r$ using a shortest path algorithm. We consider two ways of selecting objects $i^r\in \catalog$:   in the \emph{uniform} regime, $i^r$ is selected u.a.r. from the catalog $\catalog$;  in the \emph{power-law} regime, $i^r$ is selected from the catalog $\catalog$ via a power law distribution with exponent $1.2$. All  the parameter values, e.g., catalog size $|\catalog|$, number of requests $|\demand|$, number of query sources $|Q|$, and caching capacities $c_v$ are presented in Table~\ref{tab:networks}. 

We construct heterogeneous service rates as follows. Every queue service rate is either set to a low value $\mu_e=\mu_{\text{low}}$ or a high value $\mu_e=\mu_{\text{high}},$ for all $e \in E.$ We select $\mu_{\text{low}}$ and $\mu_{\text{high}}$ as follows.  Given the demands $r\in \demand$ and the corresponding arrival rates $\lambda^r$, we compute the highest load under no caching ($\mathbf{x}=\mathbf{0}$), i.e., we find  
$\lambda_{\max} = \max_{e\in E} \sum_{r:e\in p^r} \lambda^r.$
We then set $\mu_{\text{low}} = \lambda_{\max}\times 1.05$ and $\mu_{\text{high}} = \lambda_{\max}\times 200$. We  set the service rate to $\mu_{\text{low}}$ for all congested edges, i.e., edges $e$ s.t.~ $\lambda_e=\lambda_{\max}$. We set the service rate for each remaining edge $e\in E$ to $\mu_{\text{low}}$ independently with probability 0.7, and to $\mu_{\text{high}}$ otherwise.  Note that, as a result $\textbf{0}\in \domain_{\bm{\lambda}}=\domain$, i.e., the system is stable even in the absence of caching and, on average, 30 percent of the edges have a high service rate.

\begin{figure}[!t]
\vspace{-1mm}
\centering 
     \subfloat[\texttt{dtelekom}]{\includegraphics[width=0.5\columnwidth]{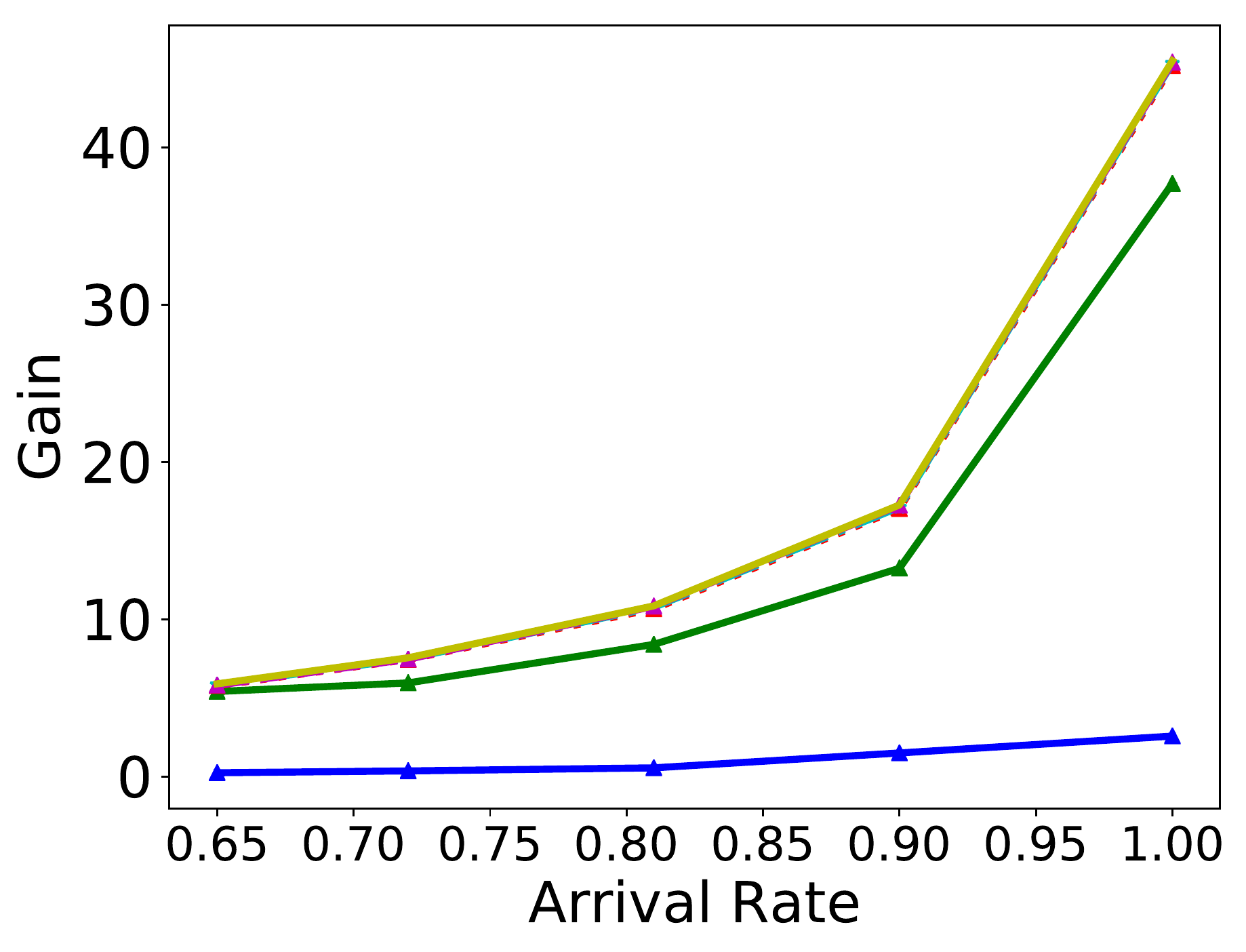}}
     \subfloat[\texttt{ER}]{\includegraphics[width=0.5\columnwidth]{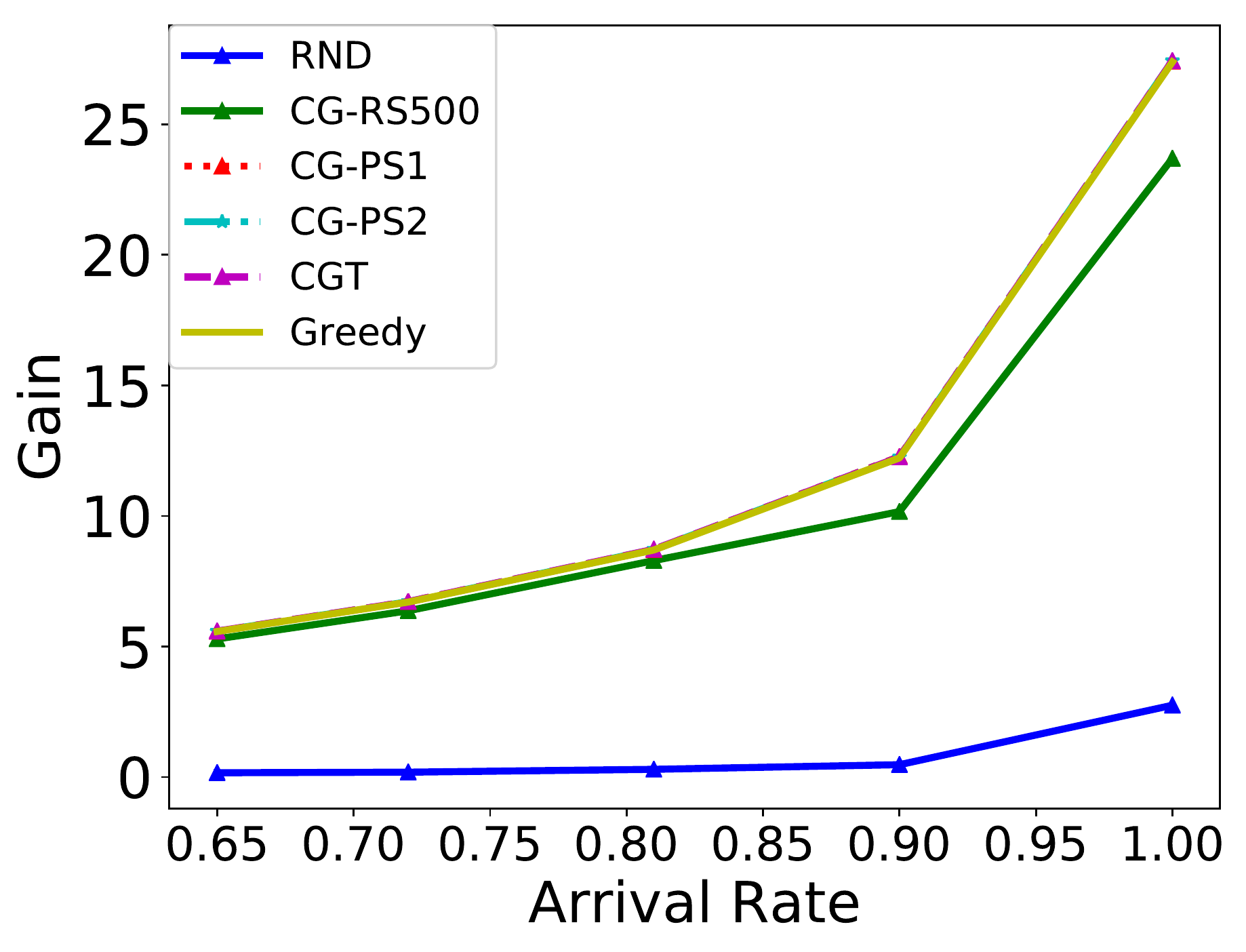}}
\caption{Caching gain vs. arrival rate. As the arrival rate increases caching gains get larger.}\label{fig:rates}
\vspace{-1mm}
\end{figure}

\begin{figure}[!t]
\vspace{-1mm}
\centering
     \subfloat[\texttt{dtelekom}]{\includegraphics[width=0.5\columnwidth]{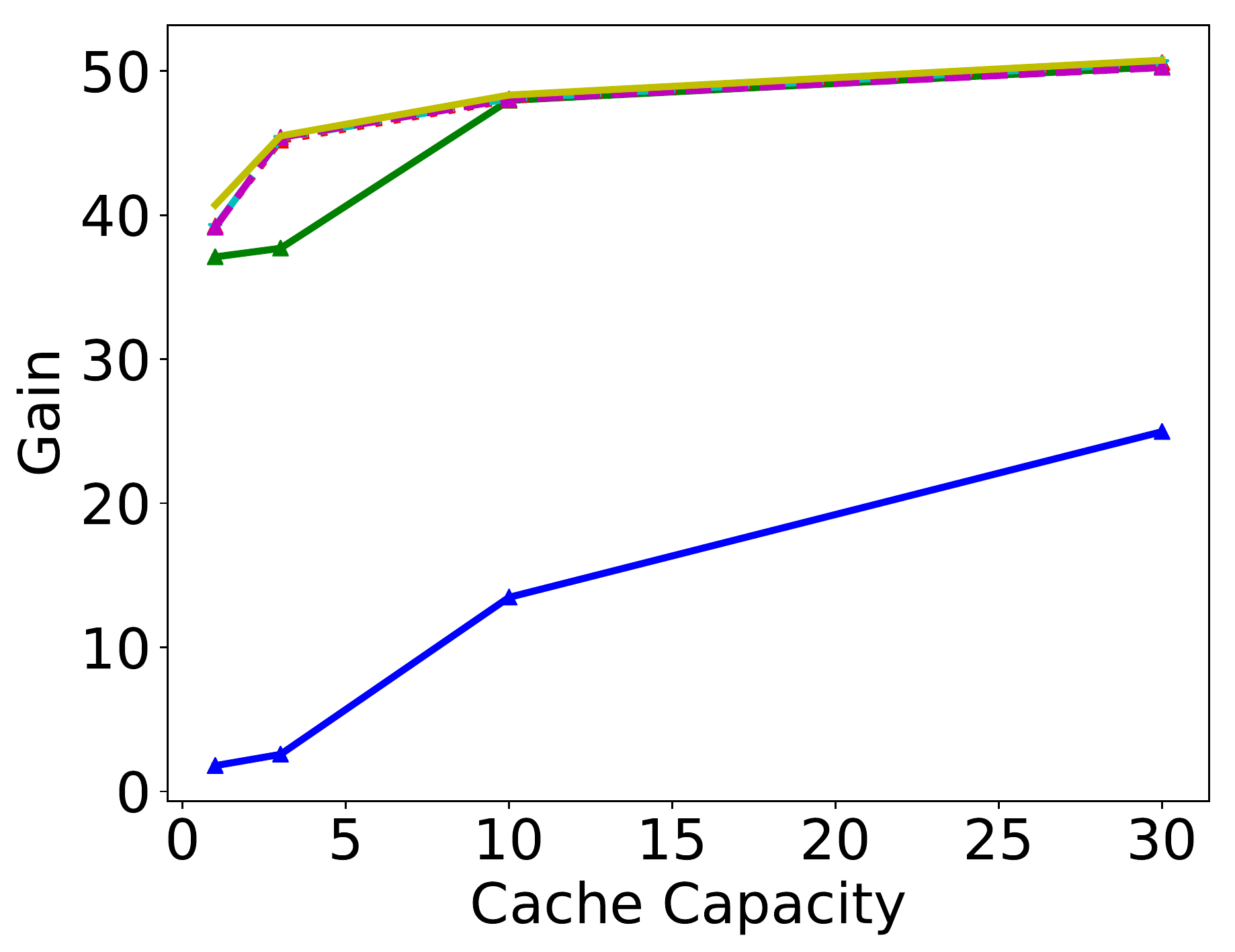}}
     \subfloat[\texttt{ER}]{\includegraphics[width=0.5\columnwidth]{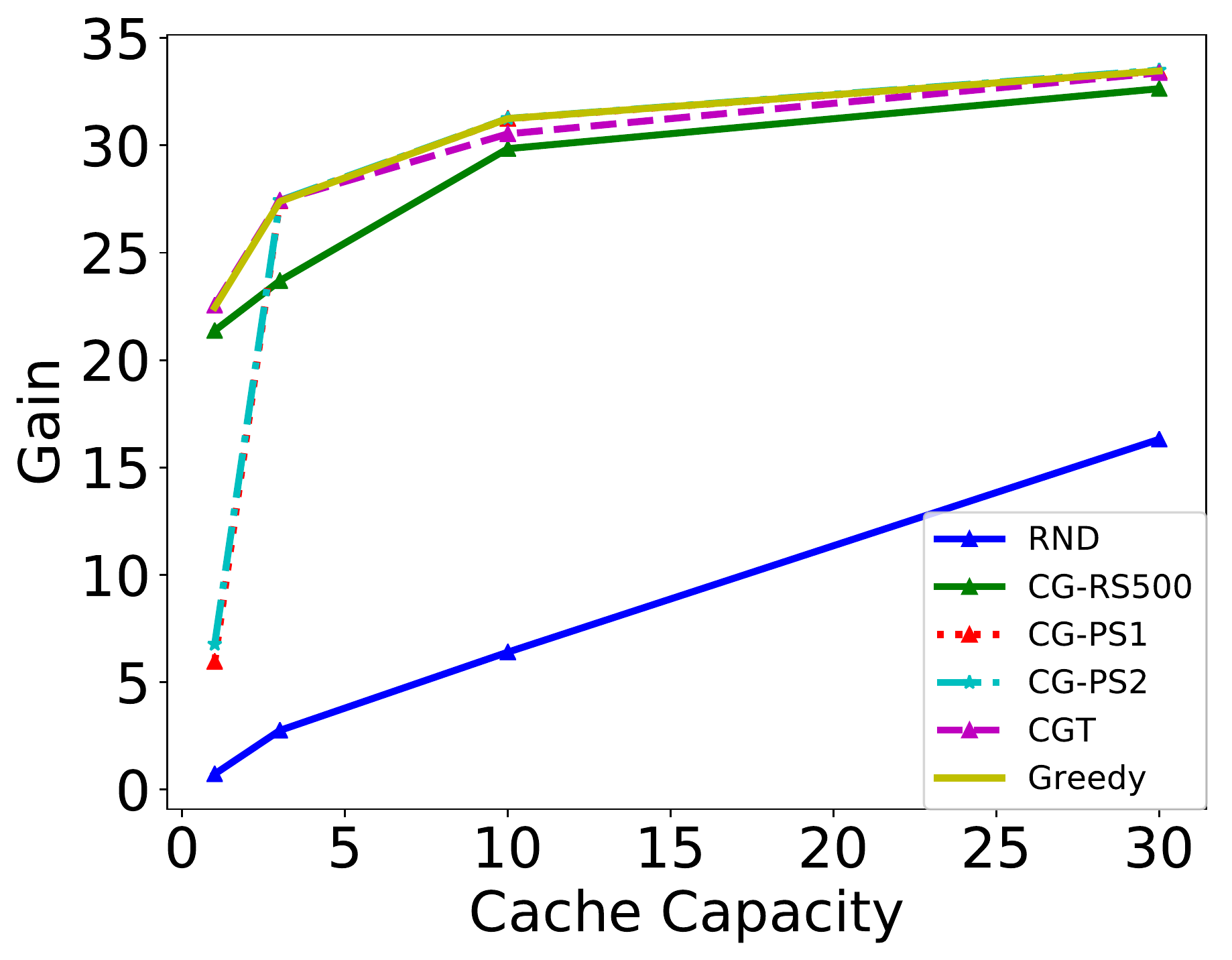}}
     \caption{Caching gain vs. cache capacity. As caching capacities increase, caching gains rise. }\label{fig:cap}
\end{figure}

\noindent\textbf{Placement Algorithms.} We implement several placement algorithms: (a) \emph{Greedy}, i.e., the greedy algorithm (Alg.~\ref{alg:greedy}), (b) \emph{Continuous-Greedy with Random Sampling} (CG-RS), i.e., Algorithm~\ref{alg:cg} with a gradient estimator based on sampling, as described in Sec.~\ref{sec:sample}, (c) \emph{Continuous-Greedy with Taylor approximation} (CGT), i.e.,    Algorithm~\ref{alg:cg} with a gradient estimator based on the Taylor expansion, as described in Sec.~\ref{sec:taylor}, and (d) \emph{Continuous-Greedy with Power Series approximation} (CG-PS), i.e., Algorithm~\ref{alg:cg} with a gradient estimator based on the power series expansion, described also in Sec.~\ref{sec:taylor}. In the case of CG-RS, we collect 500 samples, i.e., $T=500$. In the case of CG-PS we tried the first and second order expansions of the power series as CG-PS1 and CG-PS2, respectively. In the case of CGT, we tried the first-order expansion $(L=1)$. In both cases, subsequent to the execution of Alg.~\ref{alg:cg} we produce an integral solution in $\domain$ by rounding via the swap rounding method \cite{chekuri2010dependent}. All continuous-greedy algorithms use $\gamma=0.001.$ We also implement a random selection algorithm (RND), which caches $c_v$ items at each node $v \in V$, selected uniformly at random. We repeat RND 10 times, and report the average running time and caching gain. 

\noindent\textbf{Caching Gain Across Different Topologies.} 
The caching gain $F(\mathbf{x})$ for $\mathbf{x}$ generated by different placement algorithms, is shown for power-law arrival distribution and uniform arrival distribution in Figures ~\ref{cache_gain_plaw} and \ref{cache_gain_uniform}, respectively. The values are normalized by the gains obtained by RND, reported in Table~\ref{tab:networks}. Also, the running times of the algorithms for power-law arrival distribution are reported in Fig.~\ref{fig:time}.
As we see in Fig.~\ref{cache_gain_1}, Greedy is comparable to other algorithms in most topologies. However, for topologies \texttt{path} and \texttt{abilene} Greedy obtains a sub-optimal solution, in comparison to the continuous-greedy algorithm. In fact, for \texttt{path} and \texttt{abilene} Greedy performs even worse than RND. 
In Fig.~\ref{cache_gain_1}, we see that the continuous-greedy algorithms with gradient estimators based on Taylor and Power series expansion, i.e., CG-PS1, CG-PS2, and CGT outperform CG-RS500 in most topologies. Also, from Fig.~\ref{fig:time}, we see that CG-RS500 runs 100 times slower than the continuous-greedy algorithms with first-order gradient estimators, i.e., CG-PS1 and CGT. Note that 500 samples are significantly below the value\techrep{, stated in Theorem~\ref{thrm:sampling},}{} needed to attain the theoretical guarantees of the continuous-greedy algorithm, which is quadratic in $|V||\catalog|$.

\noindent\textbf{Varying Service Rates.}
For topologies \texttt{path} and \texttt{abilene}, the approximation ratio of Greedy is $\approx 0.5.$ This ratio  is a function of service rate of the high-bandwidth link $M.$  In this experiment, we explore the effect of varying $M$ on the performance of the algorithms in more detail. We plot the caching gain obtained by different algorithms for \texttt{path} and \texttt{abilene} topologies, using different values of $M\in \{M_{\min},10,20,200\},$ where $M_{\min}$ is the value that puts the system on the brink of instability, i.e., 1 and $2+\epsilon$ for \texttt{path} and \texttt{abilene}, respectively. Thus, we gradually increase the discrepancy between the service rate of low-bandwidth and high-bandwidth links. 
The corresponding caching gains are plotted in Fig.~\ref{fig:M}, as a function of $M$. We see that as $M$ increases the gain attained by Greedy worsens in both topologies: when $M=M_{\min}$ Greedy matches the performance of the continuous-greedy algorithms, in both cases. However, for higher values of $M$ it is beaten not only by all variations of the continuous-greedy algorithm, but by RND as well.

\noindent\textbf{Effect of Congestion on Caching Gain.}
In this experiment, we study the effect of varying arrival rates $\lambda^r$ on caching gain $F$. We report results only for the \texttt{dtelekom} and \texttt{ER} topologies and power-law arrival distribution. We obtain the cache placements $\mathbf{x}$  using the parameters presented in Table~\ref{tab:networks} and different arrival rates: $\lambda^r\in\{0.65, 0.72, 0.81, 0.9, 1.0\},$ for $r \in \demand$. Fig.~\ref{fig:rates} shows the caching gain attained by the placement algorithms as a function of arrival rates. We observe that as we increase the arrival rates, the caching gain attained by almost all algorithms, except RND, increases significantly. Moreover, CG-PS1, CG-PS2, CGT, and Greedy have a similar performance, while CG-RS500 achieves lower caching gains. 

\noindent\textbf{Varying Caching Capacity.}
In this experiment, we study the effect of increasing cache capacity $c_v$ on the acquired caching gains. Again, we report the results only for the \texttt{dtelekom} and \texttt{ER} topologies and power-law arrival distribution. We evaluate the caching gain obtained by different placement algorithms using the parameters of Table~\ref{tab:networks} and different caching capacities: $c_v\in\{1,3, 10,30\}$ for $v \in V.$ The  caching gain is plotted in Fig.~\ref{fig:cap}. As we see, in all cases the obtained gain increases, as we increase the caching capacities. This is expected: caching more items  reduces traffic and delay,  increasing the gain.

 	\section{Conclusions}\label{sec:conclusions}
	Our analysis suggests feasible object placements targeting many design objectives of interest, including system size  and delay, can be determined using combinatorial techniques. Our work leaves the exact characterization of approximable objectives for certain classes of queues, including M/M/1/k queues, open. 
Our work also leaves open problems relating to stability. This includes the characterization of the stability region of arrival rates $\Lambda = \cup_{\mathbf{x}\in \domain} \Lambda(\mathbf{x})$.  It is not clear whether determining membership in this set (or, equivalently, given $\lambda$, determining whether  there exists a $\mathbf{x}\in \domain$ under which the system is stable) is NP-hard or not, and whether this region can be somehow approximated. Finally, all algorithms presented in this paper are offline: identifying how to determine placements in an online, distributed fashion, in a manner that attains a design objective (as in \cite{ioannidis2016adaptive,ioannidis2017jointly}), or even stabilizes the system (as in \cite{yeh2014vip}), remains an important open problem.

        \section{Acknowledgements}

        The authors gratefully acknowledge support from National Science Foundation grant  NeTS-1718355, as well as from research grants by Intel Corp.~and Cisco
Systems.

   \techrep{
	\appendices
\section{Kelly Networks}\label{app:classickelly}
Kelly networks  \cite{kelly,nelson,gallager-stochastic} (i.e., multi-class Jackson networks) are networks of queues operating under a fairly general service disciplines (including FIFO, LIFO, and processor sharing, to name a few). As illustrated in Fig.~\ref{fig:kelly-classic}(a),  a Kelly network can be represented by a directed graph $G(V,E)$, in which each edge is associated with a queue. In the case of First-In First-Out (FIFO) queues, each edge/link $e\in E$ is associated with an M/M/1 queue  with service rate $\mu_e\geq 0$. In an \emph{open} network, packets of exogenous traffic arrive, are routed through consecutive queues, and subsequently exit the network; the path followed by a packet is determined by its {\em class}.


Formally,  let $\demand$ be the set of packet classes. For each packet class $r\in \demand$, we denote by $p^r \subseteq V$ the simple path of adjacent nodes 
 visited by a packet. Packets of class $r\in \demand$ arrive according to an exogenous Poisson arrival process with rate $\lambda^r>0$, independent of other arrival processes and service times. Upon arrival, a packet travels across nodes in $p^r$,  traversing intermediate queues, and exits upon reaching the terminal node in $p^r$.

A Kelly network forms a Markov process over the state space determined by queue contents. In particular, let $n_e^r$ be the number of packets of class $r\in \demand$ in queue $e\in E$, and $n_e = \sum_{r\in \demand} n_e^r$ be the total queue size. The state of a queue $\mathbf{n}_{e}\in \demand^{n_{e}}$, $e\in E$, is the vector of length $n_{e}$ representing the class of each packet in each position of the queue. The \emph{system state} is then given by $\mathbf{n}=[\mathbf{n}_{e}]_{e\in E}$;  we denote by $\Omega$ the state space of this Markov process.

The aggregate arrival rate $\lambda_{e}$ at an edge $e\in E$ is given by $\lambda_{e} = \sum_{r:e\in p^r}\lambda^r$, while the \emph{load} at edge $e\in E$ is given by $\rho_{e}=\lambda_{e}/\mu_{e}$. Kelly's extension of Jackson's theorem \cite{kelly} states that,  if $\rho_{e}<1$ for all $e\in E$, the Markov process $\{\mathbf{n}(t);t\geq 0\}_{t\geq 0}$ is positive recurrent, and its steady-state distribution has the following \emph{product form}:
\begin{align}
\pi(\mathbf{n}) =  \prod_{e\in E} \pi_e(\mathbf{n}_e), \quad \mathbf{n}\in \Omega,\label{steadystatek}
\end{align}  
where
\begin{align}
\pi_e(\mathbf{n}_e) = (1-\rho_e) \prod_{r\in \demand:e\in p^r} \left(\frac{\lambda^r }{\mu_e}\right)^{n_e^r}.\label{marginals}
\end{align}
As a consequence, the  queue sizes  $n_e$, $e\in E$, also have a product form distribution in steady state, and their marginals are given by: \begin{align}\prob[n_e = k] = (1-\rho_e)\rho_e^k, \quad k \in \naturals.\label{queuesize} \end{align}

The steady-state distribution  \eqref{steadystatek} holds for many different service principles beyond FIFO (c.f.~Section 3.1 of \cite{kelly} and Appendix~\ref{app:kelly}). In short, incoming packets can be placed in random position within the queue according to a given distribution, and the (exponentially distributed) service effort can be split across different positions, possibly unequally;  both placement and service effort distributions are class-independent. This captures  a broad array of policies including FIFO, Last-In First-Out (LIFO), and processor sharing: in all cases, the steady-state distribution is given by \eqref{steadystatek}.

\section{Monotone Separable Costs}\label{app:convexity}

Consider the state-dependent cost functions $c_e:\Omega\to\reals_+$ introduced in Section~\ref{sec:mincost}. The cost at state $\mathbf{n}\in \Omega$ can be written as
$c(\mathbf{n}) = \sum_{e\in E} c_e(n_e).$ Hence
$\expect[c(\mathbf{n})] =   \sum_{e\in E} \expect[c_e(n_e)]. $
On the other hand, as $c_e(n_e)\geq 0$, we have that
\begin{align}
\expect[c_e(n_e)] &= \sum_{n=0}^{\infty} c_e(n)  \prob(n_e = n)\nonumber  \\
&=c_e(0) +\sum_{n=0}^{\infty}(c_e(n+1)-c_e(n)) \prob(n_e > n)\nonumber\\
& \stackrel{\eqref{queuesize}}{=} c_e(0)+ \sum_{n=0}^{\infty}(c_e(n+1)-c_e(n)) \rho_e^n \label{eq:powerseries}
\end{align}
As $c_e$ is non-decreasing, 
$c_e(n+1)-c_e(n)\geq 0$ for all $n\in \naturals$. On the other hand, for all $n\in \naturals$, $\rho^n$ is a convex non-decreasing function of $\rho$ in $[0,1)$, so $\expect[c_e(n_e)]$ is a convex function of $\rho$ as a positively weighted sum of convex non-decreasing functions. \qed 

\section{Proof of Theorem~\ref{thm:coststruct} }
\label{app:coststruct}
We first prove the following auxiliary lemma:
\begin{lemma}\label{lemma:supermodularity}
	Let $f:\reals\to\reals$ be a convex and non-decreasing function. Also, let $g: \mathcal{X}\to \reals$ be a non-increasing supermodular set function. Then $h(\mathbf{x})\triangleq f(g(\mathbf{x}))$ is also supermodular.
\end{lemma}	  
\begin{proof}
	
	Since $g$ is non-increasing, for any $\mathbf{x},\mathbf{x}' \subseteq \mathcal{X}$ we have
	$$g(\mathbf{x}\cap \mathbf{x}')\geq g(\mathbf{x})\geq g(\mathbf{x}\cup \mathbf{x}'),$$
	$$g(\mathbf{x}\cap \mathbf{x}')\geq g(\mathbf{x}')\geq g(\mathbf{x}\cup \mathbf{x}').$$
	Due to supermodularity of $g$, we can find $ \alpha, \alpha' \in [0,1]$, $ \alpha+ \alpha'\leq 1$ such that
	$$g(\mathbf{x})= (1- \alpha)g(\mathbf{x}\cap \mathbf{x}')+ \alpha g(\mathbf{x}\cup \mathbf{x}'),$$
	$$g(\mathbf{x}')= (1- \alpha')g(\mathbf{x}\cap \mathbf{x}')+ \alpha' g(\mathbf{x}\cup \mathbf{x}').$$
	
	Then, we have
	\begin{IEEEeqnarray*}{+rCl+x*}
		f(g(\mathbf{x}))&+&f(g(\mathbf{x}'))\\ &\leq&   (1- \alpha) f(g(\mathbf{x}\cap \mathbf{x}'))+  \alpha f(g(\mathbf{x}\cup \mathbf{x}')) \\&+& (1- \alpha')f(g(\mathbf{x}\cap \mathbf{x}'))+ \alpha'f(g(\mathbf{x}\cup \mathbf{x}'))\\
		&=& f(g(\mathbf{x}\cap \mathbf{x}'))+f(g(\mathbf{x}\cup \mathbf{x}'))\\&+& (1- \alpha - \alpha')(f(g(\mathbf{x}\cap \mathbf{x}'))-f(g(\mathbf{x}\cup \mathbf{x}')))\\
		&\leq&  f(g(\mathbf{x}\cap \mathbf{x}'))+f(g(\mathbf{x}\cup \mathbf{x}')),
	\end{IEEEeqnarray*}
	where the first inequality is due to convexity of $f$, and the second one is because $ \alpha+ \alpha'\leq 1$ and $f(g(.))$ is non-increasing. This proves $h(\mathbf{x})\triangleq f(g(\mathbf{x}))$ is supermodular.
\end{proof}
To conclude the proof of Thm.~\ref{thm:coststruct}, observe that
it is easy to verify that $\rho_{e}, \forall e\in E$, is supermodular and non-increasing in $S$. Since, by Assumption \ref{as:conv}, $C_e$ is a non-decreasing function, then, $C_e(S)\triangleq C_{e}(\rho_{u,v}(S))$ is non-increasing. By Lemma \ref{lemma:supermodularity}, $C_{s}(S)$ is also supermodular. Hence, the cost function is non-increasing  and supermodular as the sum of non-increasing and supermodular functions.

\section{Proof of Lemma~\ref{lem:tight}}\label{app:proofoflem:tight}
Consider the path topology illustrated in Fig.~\ref{fig:halfistight}. Assume that requests for files 1 and 2 are generated at node $u$ with rates $\lambda_1=\lambda_2=\delta$, for some  $\delta\in(0,1)$. Files 1 and 2 are stored permanently at $v$ and $z$, respectively. Caches exist only on $u$ and $w$, and have capacity $c_u=c_w=1$. Edges $(u,v)$, $(w,z)$ have bandwidth $\mu_{(u,v)}=\mu_{(w,z)}=1$, while edge $(u,w)$ is a high bandwidth link, having capacity $M\gg 1$. Let $\mathbf{x}_0=\mathbf{0}$.
The greedy algorithm starts from empty caches  and adds item 2 at cache $u$. This is because the caching gain from this placement is $c_{(u,w)}+c_{(w,z)}= \frac{1}{M-\delta} + \frac{1}{1-\delta}$, while the caching gain of all other decisions is at most $\frac{1}{1-\delta}$. Any subsequent caching decisions do not change the caching gain. The optimal solution is to cache item 1 at $u$ and item 2 at $w$, yielding a caching gain of $2/(1-\delta)$. Hence, the greedy solution attains an approximation ratio 
$0.5\cdot (1+\frac{1-\delta}{M-\delta}).$
By  appropriately choosing $M$ and $\delta$, this can be made arbitrarily close to 0.5. \qed

	\section{Rounding}\label{sec:rounding}

	Several poly-time algorithms can be used to round the fractional solution that is produced by Alg.~\ref{alg:cg} to an integral $\mathbf{x}\in \domain$.	
     We briefly review two such rounding algorithms: pipage rounding \cite{ageev2004pipage}, which is deterministic, and swap-rounding~\cite{chekuri2010dependent}, which is randomized.
For a more rigorous treatment, we refer the reader to \cite{ageev2004pipage,ioannidis2016adaptive} for pipage rounding, and \cite{chekuri2010dependent} for swap rounding.

 Pipage rounding uses the following property of $G$: given a fractional solution $\mathbf{y} \in \tilde{\domain}$, there are at least two fractional variables $y_{vi}$ and $y_{v'i'}$, such that transferring mass from one to the other,$1)$ makes at least one of them 0 or 1, $2)$ the new $\hat{\mathbf{y}}$ remains feasible in $\tilde{\domain}$, and $3)$ $G(\hat{\mathbf{y}}) \geq G(\mathbf{y}(1))$, that is, the expected caching gain at $\hat{\mathbf{y}}$ is at least as good as $\mathbf{y}$. This process is repeated until $\hat{\mathbf{y}}$ does not have any fractional element, at which point pipage rounding terminates and return $\hat{\mathbf{y}}$.
	This procedure has a run-time of $O(|V||\mathcal{C}|)$ \cite{ioannidis2016adaptive}, and since
	each rounding step can only increase $G$, it follows that the final integral  $\hat{\mathbf{y}}\in \domain$ must satisfy
	\begin{equation*}
	F(\hat{\mathbf{y}})=G(\hat{\mathbf{y}})\geq \mathbb{E}[G(\mathbf{y})]  \geq (1-\frac{1}{e})G(\mathbf{y}^*) \geq (1-\frac{1}{e}) F(\mathbf{x}^*),
	\end{equation*}
       where $\mathbf{x}^*$ is an optimal solution to \textsc{MaxCG}. Here, the first equality holds because $F$ and $G$ are equal when their arguments are integral, while the last inequality holds because \eqref{eq:nonconv} is a relaxation of \textsc{MaxCG}, maximizing the same objective over a larger domain. 	

In swap rounding, given a fractional solution $\mathbf{y} \in \tilde{\domain}$ produced by Alg.~\ref{alg:cg} observe that it can be written as a convex combination of integral vectors in $\domain$, i.e., $\mathbf{y} = \sum_{k=1}^K \gamma_k \mathbf{m}_k,$ where $\gamma_k\in [0,1], \sum_{k=1}^K\gamma_k=1,$ and $\mathbf{m}_k \in\domain$.Moreover, by construction, each such vector $\mathbf{m}_k$ is maximal, in that all capacity constraints are tight. Swap rounding iteratively merges these constituent integral vectors, producing an integral solution. At each iteration $i$, the present integral vector $\mathbf{c}_k$ is merged with $\mathbf{m}_{k+1}\in\mathcal{D}$ into a new integral solution $\mathbf{c}_{k+1}\in \mathcal{D}$ as follows:  if the two solutions $\mathbf{c}_k$, $\mathbf{m}_{k+1}$ differ at a cache $v\in V$, items in this cache are swapped to reduce the set difference: either an item $i$ in a cache in $\mathbf{c}_k$  replaces an item  $j$ in $\mathbf{m}_{k+1}$, or an item  $j$ in $\mathbf{m}_{k+1}$ replaces an item $i$ in $\mathbf{c}_k$; the former occurs with probability proportional to   $\sum_{\ell=1}^{k}\gamma_{\ell}$, and the latter with probability proportional to $\gamma_{k+1}$. The swapping is repeated until the two integer solutions become identical; this merged solution becomes $\mathbf{c}_{k+1}$. This process terminates after $K-1$ steps, after which all the points $\mathbf{m}_k$ are merged into a single integral vector $\mathbf{c}_K\in \mathcal{D}.$ Observe that, in contrast to pipage rounding, swap rounding does not require any evaluation of the objective $F$ during rounding. This makes swap rounding significantly faster to implement; this comes at the expense of the approximation ratio, however, as the resulting guarantee $1-1/e$ is in expectation.

\section{Proof of Lemma~\ref{wdnfmoments}}\label{app:wdnfmomentslemma} 
We prove this by induction on $k\geq 1$. 
Observe first that, by \eqref{eq:rho_X}, the load on each edge $e=(u,v)\in E$ can be written as a polynomial of the following form:
\begin{align}
	\rho_e(\mathbf{x},\bm{\lambda}) = \sum_{r\in \demand_e}\beta_r(\bm{\lambda}) \cdot \prod_{j\in \mathcal{I}_e(r)}(1-x_j)\label{rhoiswdnf},
\end{align}
for appropriately defined 
\begin{align*}\demand_e &=\demand_{(u,v)} = \{r\in \demand: (v,u)\in p^r \},\\
\mathcal{I}_e(r) &=\{(w,i^r)\in V\times \catalog: w\in p^r, k_{p^r}(w)\leq k_{p^r}(v)\}, \text{ and}\\
  \beta_r(\bm{\lambda})&=\lambda^r/\mu_e.
\end{align*}

In other words, $\rho_e:\domain_{\bm{\lambda}}\to \reals$ is indeed a \emph{W-DNF} polynomial.  
For the induction step, observe that W-DNF polynomials, seen as functions over the integral domain $\domain_{\bm{\lambda}}$, are \emph{closed under multiplication}. In particular, the following lemma holds:
\begin{lemma}\label{lemma:prodclosed}
	Given two W-DNF  polynomials $f_1:\domain_{\bm{\lambda}}\to\reals$ and $f_2:\domain_{\bm{\lambda}}\to\reals$, given by 
\begin{align*}
f_1(\mathbf{x}) &= \sum_{r\in \demand_1}\beta_r \prod_{t\in \mathcal{I}_1(r)}(1-x_t), \quad \text{and}\\
f_2(\mathbf{x}) &= \sum_{r\in \demand_2}\beta_r \prod_{t\in \mathcal{I}_2(r)}(1-x_t),\end{align*} 
their product $f_1\cdot f_2$ is also a W-DNF polynomial over $\domain_{\bm{\lambda}}$, given by:
$$(f_1 \cdot f_2)(\mathbf{x}) = \sum_{(r,r')\in \demand_1\times\demand_2} \beta_r\beta_r' \prod_{t\in \mathcal{I}_1(r)\cup\mathcal{I}_2(r')} (1-x_t)$$
\end{lemma}
\begin{proof}
To see this, observe that 
\begin{align*}
&f_1(\mathbf{x}) f_1(\mathbf{x}) =  \\&\sum_{(r,r')\in \demand_1\times\demand_2} \!\!\!\beta_r\beta_r'\!\!\!\!\! \prod_{t\in \mathcal{I}_1(r)\cap\mathcal{I}_2(r')}\!\!\! (1-x_t)^2\!\!\!\!\! \prod_{t\in \mathcal{I}_1(r)\triangle\mathcal{I}_2(r')}\!\!\!(1-x_t)
\end{align*}
where $\triangle$ is the symmetric set difference. On the other hand, as $(1-x_t)\in \{0,1\}$, we have that $(1-x_t)^2=(1-x_t)$, and the lemma follows.
\end{proof}
Hence, if  $\rho_e^k(\mathbf{x},\bm{\lambda})$ is a W-DNF polynomial, by \eqref{rhoiswdnf} and Lemma~\ref{lemma:prodclosed}, so is  $\rho_e^{k+1}(\mathbf{x},\bm{\lambda})$.\qed

\section{Proof of Lemma~\ref{taylorlemma}}\label{proofoftaylorlemma}
The Taylor expansion of $C_e$ at $\rho^*$ is given by:
\begin{align*}
	C_{e}(\rho)& = C_{e}(\rho^*) + \sum_{k=1}^{L}\frac{1}{k!}C^{(k)}_{e}(\rho^*)(\rho-\rho^*)^k + \\
&\quad+	\frac{1}{(L+1)!}C^{(L+1)}_{e}(\rho')(\rho-\rho^*)^{L+1},
\end{align*}
where $\rho' \in [\rho^*,\rho]$ and
 $C_e^{(k)}$ is the $k$-th order derivative of $C_e$. 
By expanding this  polynomial and reorganizing the terms, we get
$$	C_{e}(\rho) =  \sum_{k=0}^{L}\alpha^{(k)}_e\rho^k+ 
\frac{1}{(L+1)!}C^{(L+1)}_{e}(\rho')(\rho-\rho^*)^{L+1},$$
where 
	$$\textstyle\alpha^{(k)}_e  = \sum_{i=k}^{L} \frac{(-1)^{i-k}\binom{i}{k}}{i!}C^{(i)}_{e}(\rho^*)(\rho^*)^{i-k},$$ for $k=0,1,\cdots,L.$
Consider now the $L$-th order Taylor approximation of $C_e$, given by
$$\hat{C}_e(\rho) =  \sum_{k=0}^{L}\alpha^{(k)}_e\rho^k. $$
Clearly, this is an estimator of $C_e$, with an error of the order $|C_e(\rho)-\hat{C}_e(\rho)|=o\left((\rho-\rho_*)^L\right).$ Thus, for $\mathbf{x} \in \domain$ a random Bernoulli vector parameterized by $\mathbf{y}\in \tilde{\domain}$,
\begin{align}\!\!\!\expect_{\mathbf{y}} [C_e(\rho_e(\mathbf{x},\bm{\lambda}))] &\approx \expect_{\mathbf{y}} [\hat{C}_e(\rho_e(\mathbf{x},\bm{\lambda}))]
=  \sum_{k=0}^{L}\alpha^{(k)}_e \expect_{\mathbf{y}}[\rho^k_e(\mathbf{x},\bm{\lambda})]\label{eq:test} \end{align}
On the other hand, for all $v\in V$ and $i\in \catalog$:
\begin{align}
	\frac{\partial G(\mathbf{y})}{\partial y_{vi}}& \stackrel{\eqref{eq:deriv_G} }{=} \mathbb{E}_{\mathbf{y}}[F(\mathbf{x})|x_{vi}=1]- \mathbb{E}_{\mathbf{y}}[F(\mathbf{x})|x_{vi}=0]\nonumber\\
& \stackrel{\eqref{eq:fobj}}{=} \mathbb{E}_{\mathbf{y}}[C(\mathbf{x})|x_{vi}=0] - \mathbb{E}_{\mathbf{y}}[C(\mathbf{x})|x_{vi}=1]\nonumber\\
\begin{split}\label{eq:taylorapprox}
& \stackrel{\eqref{eq:obj},\eqref{eq:test}}{\approx} \sum_{e\in E} \sum_{k=1}^{L}\alpha^{(k)}_e \Big (\expect_{\mathbf{y}}[\rho^k_e(\mathbf{x},\bm{\lambda})|x_{vi}=0] \\&\qquad\qquad\qquad\qquad-\expect_{\mathbf{y}}[\rho^k_e(\mathbf{x},\bm{\lambda})|x_{vi}=1] \Big),
\end{split}
\end{align} 
where the error of the approximation is given by
\begin{align*}&\frac{1}{(L+1)!}\sum_{e\in E}C^{(L+1)}_{e}(\rho') \Big[ \expect_{\mathbf{y}}[(\rho_e(\mathbf{x},\bm{\lambda})-\rho^*)^{L+1}|x_{vi}=0] \\
&\quad-  \expect_{\mathbf{y}}[(\rho_e(\mathbf{x},\bm{\lambda})-\rho^*)^{L+1}|x_{vi}=1]\Big] \end{align*}
The lemma thus follows from Lemmas~\ref{wdnfmoments} and~\ref{lemma:compute}.


\section{Proof of Theorem \ref{theorem:converge}}
\label{app:3}

We begin by bounding the bias of estimator \eqref{eq:taylorapprox}.
Indeed, given a set of continuous functions $\{C_{(u,v}\}_{(u,v)\in E}$ where their first $L+1$ derivatives within their operating regime, $[0,1)$, are upperbounded by a finite constant, $W$, the bias of estimator $\mathbf{z}\equiv [z_{vi}]_{v\in V,i\in \mathcal{C}}$, where $z_{vi}$ is  defined by \eqref{eq:taylorapprox2}, is given by
\begin{align}
B&\equiv ||\mathbf{z}\ -\triangledown G(\mathbf{y})||_2\nonumber\\& =	||\sum_{e\in E}\frac{1}{(L+1)!}C^{(L+1)}_{e}(\rho'_{e})(\rho_{e}-\rho^*_{e})^{L+1}||_2, \label{eq:bias}
\end{align}
where $\rho'_{e} \in [\rho^*_{e},\rho_{e}]$.  To compute the bias, we note that $\rho_{e}, \rho^*_{e} \in [0,1]$. Specifically, we assume $\rho_{e}, \rho^*_{e} \in [0,1)$. Hence, $|\rho_{e}- \rho^*_{e}| \leq 1$, and $C^{(L+1)}_{e}(\rho'_{e}) \leq \max \{C^{(L+1)}_{e}(\rho_{e}) ,C^{(L+1)}_{e}(\rho^*_{e}) \} <\infty$. In particular, let $ W =\max_{e\in E} C^{(L+1)}_{e}(\rho'_{e})$. Then, it is easy to compute the following  upper bound on the bias of $\mathbf{z}$:
\begin{equation}\label{eq:bias_bound}
B \leq \frac{W|E|}{(L+1)!}.
\end{equation}
In addition, note that $G$ is linear in $y_{vi}$, and hence \cite{calinescu2011}: 
\begin{equation}
\begin{split}
\frac{\partial G}{\partial y_{vi}} = \mathbb{E}[F(\mathbf{x})|x_{vi}=1]- \mathbb{E}[F(\mathbf{x})|x_{vi}=0]\\= \mathbb{E}[C(\mathbf{x})|x_{vi}=0] - \mathbb{E}[C(\mathbf{x})|x_{vi}=1]\geq 0,
\end{split}
\end{equation} 
which is $\geq 0$ due to monotonicity of $F(\mathbf{x})$. It is easy to verify that $\frac{\partial^2 G}{\partial y_{vi}^2} = 0$. For $(v_1,i_1)\neq(v_2,i_2)$,  we can compute the second derivative of $G$ \cite{calinescu2011} as given by
\begin{IEEEeqnarray*}{+rCl+x*}
	\frac{\partial^2 G}{\partial y_{v_1i_1}\partial y_{v_2i_2}} &=& \mathbb{E}[C(\mathbf{x})|x_{v_1i_1}=1,x_{v_2i_2}=0]\\ &+&\mathbb{E}[C(\mathbf{x})|x_{v_1i_1}=0,x_{v_2i_2}=1] \\&-&\mathbb{E}[C(\mathbf{x})|x_{v_1i_1}=1,x_{v_2i_2}=1]\\ &-& \mathbb{E}[C(\mathbf{x})|x_{v_1i_1}=0,x_{v_2i_2}=0] \leq 0,
\end{IEEEeqnarray*}
which is $\leq 0$ due to the supermodularity of $C(\mathbf{x})$. Hence, $G(\mathbf{y})$ is component-wise concave \cite{calinescu2011} .

In additions, it is easy to see that for $\mathbf{y}\in \tilde{\domain}$, $||G(\mathbf{y})||$, $||\triangledown G(\mathbf{y})||$, and $||\triangledown^2 G(\mathbf{y})||$ are bounded by $C(\mathbf{x}_0)$, $C(\mathbf{x}_0)$ and $2C(\mathbf{x}_0)$, respectively. Consequently, $G$ and $\triangledown G$ are $P$-Lipschitz continuous, with $P=2C(\mathbf{x}_0)$. 

In the $k$th iteration of the Continuous Greedy algorithm, let $\mathbf{m}^*=\mathbf{m}^*(\mathbf{y}_k):=(\mathbf{y}^*\vee (\mathbf{y}_k+\mathbf{y}_0))-\mathbf{y}_k=(\mathbf{y}^*-\mathbf{y}_k)\vee \mathbf{y}_0 \geq \mathbf{y}_0$, where $x\vee y:=(\max\{x_i,y_i\})_i$. Since $\mathbf{m}^*\leq \mathbf{y}^*$ and $\mathcal{D}$ is closed-down, $\mathbf{m}^*\in \mathcal{D}$. Due to monotonicity of $G$, it follows
\begin{equation}
G(\mathbf{y}_k+\mathbf{m}^*)\geq G(\mathbf{y}^*). \label{eq:G(Y+M)}
\end{equation}

We introduce univariate auxiliary function $g_{\mathbf{y},\mathbf{m}}(\xi):=G(\mathbf{y}+ \xi \mathbf{m}),\xi \in [0,1], \mathbf{m}\in \tilde{\domain}$. Since $G(\mathbf{y})$ is component-wise concave, then, $g_{\mathbf{y},\mathbf{m}}(\xi)$ is concave in $[0,1]$. 
In addition, since $g_{\mathbf{y}_k,\mathbf{m}^*}(\xi) = G(\mathbf{y}_k+\xi \mathbf{m}^*)$ is concave for $\xi\in[0,1]$, it follows
\begin{equation}
\begin{split}
g_{\mathbf{y}_k,\mathbf{m}^*}(1) - g_{\mathbf{y}_k,\mathbf{m}^*}(0) =G(\mathbf{y}_k+\mathbf{m}^*)-G(\mathbf{y}_k)\\ \leq   \frac{d g_{\mathbf{y}_k,\mathbf{m}}(0)}{d \xi}\times 1 = \langle \mathbf{m}^*,\triangledown G(\mathbf{y}_k)\rangle.\label{eq:Gdiff}
\end{split}
\end{equation}

Now let $\mathbf{m}_k$ be the vector chosen by Algorithm \ref{alg:cg} in the $k$th iteration. We have
\begin{equation}\label{eq:mAndz}
\langle\mathbf{m}_k,\mathbf{z}(\mathbf{y}_k)\rangle \geq \langle \mathbf{m}^*,\mathbf{z}(\mathbf{y}_k)\rangle .
\end{equation}

For the LHS, we have 
\begin{IEEEeqnarray}{+rCl+x*}
	& &\langle \mathbf{m}_k,\mathbf{z}\rangle  = \langle \mathbf{m}_k,\triangledown G(\mathbf{y}_k)\rangle+ \langle \mathbf{m}_k,\mathbf{z}  - \triangledown G(\mathbf{y}_k)\rangle \nonumber\\& &
	\stackrel{(i)}{\leq}  \langle \mathbf{m}_k,\triangledown G(\mathbf{y}_k)\rangle+||m_k||_2
	\cdot | \mathbf{z} - \triangledown G(\mathbf{y}_k)||2 \leq\nonumber\\&&
	\langle \mathbf{m}_k,\triangledown G(\mathbf{y}_k)\rangle + DB.
\label{eq:dummy_D}
\end{IEEEeqnarray}
where $D = \max_{\mathbf{m}\in \tilde{\domain}} \|\mathbf{m}\|_2 \leq |V|\cdot\max\limits_{v\in{V}}c_v$, is the upperbound on the diameter of $\tilde{\mathcal{D}}$, $B$ is as defined in \eqref{eq:bias_bound}, and (i) follows from Cauchy-Schwarz inequality. Similarly, we have for the RHS of that \eqref{eq:mAndz}
\begin{equation}
\langle \mathbf{m}^*,\mathbf{z}(\mathbf{y}_k)\rangle \geq \langle \mathbf{m}^*,\triangledown G(\mathbf{y}_k)\rangle - DB.
\end{equation}

It follows
\begin{IEEEeqnarray}{+rCl+x*}
	& &\langle \mathbf{m}_k,\triangledown G(\mathbf{y}_k)\rangle + 2DB \geq \langle \mathbf{m}^*,\triangledown G(\mathbf{y}_k)\rangle   \nonumber\\&&\stackrel{(a)}{\geq}
G(\mathbf{y}_k+\mathbf{m}^*) - G(\mathbf{y}_k)
	\stackrel{(b)}{\geq} G(\mathbf{y}^*) - G(\mathbf{y}_k),\label{eq:dummy_2}
\end{IEEEeqnarray}
where $(a)$ follows from \eqref{eq:Gdiff}, and $(b)$ follows from \eqref{eq:G(Y+M)}.

Using the $P$-Lipschitz continuity property of $\frac{d g_{\mathbf{y}_k,\mathbf{m}_k}(\xi)}{d \xi}$ (due to $P$-Lipschitz continuity of $\triangledown G$), it is straightforward to see that
\begin{equation}
\begin{split}
-\frac{P\gamma_k^2}{2}\leq g_{\mathbf{y}_k,\mathbf{m}_k}(\gamma_k)-g_{\mathbf{y}_k,\mathbf{m}_k}(0)- \gamma_k\cdot\frac{d g_{\mathbf{y}_k,\mathbf{m}_k}(0)}{d \xi}= \\ G(\mathbf{y}_k+\gamma_k \mathbf{m}_k)-G(\mathbf{y}_k)-\gamma_k< \mathbf{m}_k , \triangledown G(\mathbf{y}_k)>,
\end{split}
\end{equation}
hence,
\begin{IEEEeqnarray}{+rCl+x*}
	&&G(\mathbf{y}_{k+1})-G(\mathbf{y}_k) \geq \gamma_k \langle \mathbf{m}_k , \triangledown G(\mathbf{y}_k)\rangle -\frac{P\gamma_k^2}{2} \geq \nonumber\\&&
	\gamma_k\langle \mathbf{m}_k , \triangledown G(\mathbf{y}_k)\rangle-\frac{P\gamma_k^2}{2} 
	\stackrel{(c)}{\geq} \nonumber\\&&\gamma_k (G(\mathbf{y}^*) - G(\mathbf{y}_k))- 2\gamma_kDB-\frac{P\gamma_k^2}{2},\label{eq:dummy_3}
\end{IEEEeqnarray}
where $(c)$ follows from \eqref{eq:dummy_2}, respectively. By rearranging the terms and letting $k=K-1$, we have
\begin{IEEEeqnarray*}{+rCl+x*}
	&&G(\mathbf{y}_{K})-G(\mathbf{y}^*)\\&& \geq  \prod_{j=0}^{K-1}(1- \gamma_j)(G(\mathbf{y}_0) - G(\mathbf{y}^*)) - 2DB\sum_{j=0}^{K-1}\gamma_j -\frac{P}{2}\sum_{j=0}^{K-1}\gamma_j^2\\&&
	\stackrel{(e)}{\geq}(G(\mathbf{y}_0)-G(\mathbf{y}^*))\exp\{-\sum_{j=0}^{K-1}\gamma_j\} - 2DB\sum_{j=0}^{K-1}\gamma_j -\frac{P}{2}\sum_{j=0}^{K-1}\gamma_j^2,
\end{IEEEeqnarray*}
where $(e)$ is true since  $1-x\leq e^{-x}, \forall x\geq 0$, and $G(\mathbf{y}_0)\leq G(\mathbf{y}^*)$ holds due to the greedy nature of Algorithm \ref{alg:cg} and monotonicity of $G$. In addition, Algorithm \ref{alg:cg} ensures $\sum_{j=0}^{K-1}\gamma_j=1$. It follows
\begin{equation}
G(\mathbf{y}_{K})-(1-\frac{1}{e})G(\mathbf{y}^*) \geq e^{-1}G(\mathbf{y}_0)- 2DB- \frac{P}{2}\sum_{j=0}^{K-1}\gamma_j^2.\label{eq:conv_gamma}
\end{equation} 

This result holds for general stepsizes $0<\gamma_j\leq 1$. The RHS of \eqref{eq:conv_gamma} is indeed maximized when $\gamma_j= \frac{1}{K}$, which is the assumed case in Algorithm \ref{alg:cg}. In addition, we have $\mathbf{y}_0 = \mathbf{0}$, and hence, $G(\mathbf{y}_0)=0$.  Therefore, we have
\begin{equation}
G(\mathbf{y}_{K})-(1-\frac{1}{e})G(\mathbf{y}^*) \geq - 2DB - \frac{P}{2K}.
\end{equation} 
\section{General Kelly Networks}
\label{app:kelly}
In Kelly's network of queues (see Section 3.1 of \cite{kelly} for more information), queue $e\in \{1,2,\cdots,|E|\}$, assuming it contains $n_e$ packets in the queue, operates in the following manner:
\begin{enumerate}
	\item Each packet (customer) requires an exponentially distributed amount of service.
	\item A total service effort is provided by queue $e$ at the rate $\mu_e(n_e)$.
	\item The packet in position $l$ in the queue is provided with a portion $\gamma_e(l,n_e)$  of the total service effort, for $l=1,2,\cdots,n_e$; when this packet completes service and leaves the queue, packets in positions $l+1,l+2,\cdots,n_e$ move down to positions $l,l+1,\cdots,n_e-1$, respectively.
	\item An arriving packet at queue $j$ moves into position $l$, for $l=1,2,\cdots,n_e$, with probability $\delta_e(l,n_e+1)$; packets that where in positions $l,l+1,\cdots,n_e+1$, move up to positions $l+1,l+2,\cdots,n_e+1$, respectively.
\end{enumerate} 
Clearly, we require $\mu_e(n_e)>0$ for $n_e>0$; in addition,
\begin{equation}
\sum_{l=1}^{n_e} \gamma_e(l,n_e) = 1,
\end{equation}
\begin{equation}
\sum_{l=1}^{n_e} \delta_e(l,n_e) = 1.
\end{equation}

Kelly's theorem \cite{kelly} states that,  if $\rho_{e}<1$ for all $e\in E$, the state of queue $e$ in equilibrium is independent of the rest of the system, hence, it will have a product form. In addition, the probability that queue $e$ contains $n_e$ packets is
\begin{align}\label{eq:steady_gen}
\pi_e(n_e) = b_e  \frac{\lambda_e^{n_e} }{\prod_{l=1}^{n_e}\mu_e(l)},
\end{align}
where $b_e$ is the normalizing factor. As can be seen from \eqref{eq:steady_gen}, note that the steady-state distribution is not function of  $\gamma_e$'s, and $\delta_e(l,n_e+1)$'s, and hence, is independent of the packet placement and service allocation distributions.

We note that by allowing $\mu_e(l) = mu_e$, we obtain the results in \eqref{queuesize}.
\section{Proof of Lemma \ref{lemma:mmk}}
\label{app:mmk}
For an arbitrary network of M/M/k queues, the traffic load on queue $(u,v)\in {E}$ is given as
\begin{equation}
a_{(u,v)}(\mathbf{x}) = \frac{\sum\limits_{r\in\mathcal{R}: (v,u)\in p^r}\lambda^{r} \prod\limits_{k'=1}^{k_{p^r}(v)}(1-x_{p^r_{k'}i^r})}{k\mu_{(u,v)}}, \label{eq:rho_X_k}
\end{equation} 
which is similar to that of M/M/1 queues, but normalized by the number of servers, $k$. Hence, $a_{(u,v)}(\mathbf{x}) $ is submodular in $\mathbf{x}$. For an M/M/k queue, the probability that an arriving packet finds all servers busy and will be forced to wait in queue is given by Erlang C formula \cite{datanetworks}, which follows
\begin{equation}
P_{(u,v)}^{Q}(\mathbf{x}) =\frac{b_{(u,v)}(\mathbf{x})(ka_{(u,v)}(\mathbf{x}))^k}{k!(1-a_{(u,v)}(\mathbf{x}))},
\end{equation} 
where 
\begin{equation}
b_{(u,v)}(\mathbf{x}) = \left[\sum_{n=0}^{k-1}\frac{(ka_{(u,v)}(\mathbf{x}))^n}{n!}+\frac{(ka_{(u,v)}(\mathbf{x}))^k}{k!(1-a_{(u,v)}(\mathbf{x}))}\right]^{-1},
\end{equation} 
is the normalizing factor. In addition, the expected number of packets waiting for or under transmission is given by 
\begin{equation}
\mathbb{E} [n_{(u,v)}(\mathbf{x})] = ka_{(u,v)}(\mathbf{x}) + \frac{a_{(u,v)}(\mathbf{x})P_{(u,v)}^{Q}(\mathbf{x})}{1-a_{(u,v)}(\mathbf{x})}.
\end{equation}

Lee and Cohen in \cite{erlangC}, shows that $P_{(u,v)}^{Q}(\mathbf{x})$ and $\mathbb{E} [n_{(u,v)}(\mathbf{x})] $ are strictly increasing and convex in $a_{(u,v)}(\mathbf{x})$, for $a_{(u,v)}(\mathbf{x})\in [0,1)$. In addition, a more direct proof of convexity of $\mathbb{E} [n_{(u,v)}(\mathbf{x})] $ was shown by Grassmann in \cite{MMK}. Hence, Both $P(\mathbf{x}) := \sum_{(u,v) \in {E}}P_{(u,v)}^{Q}(\mathbf{x})$ and $N(\mathbf{x}) := \sum_{(u,v) \in {E}}\mathbb{E} [n_{(u,v)}(\mathbf{x})] $ are increasing and convex. Due to Theorem \ref{thm:coststruct}, we note that both functions are non-increasing and supermodular in $\mathbf{x}$, and the proof is complete.
\section{Networks of Symmetric Queues}
\label{app:symmetric}
Let $n_e$ be the number of packets placed in positions $1,2,\cdots,n$ in queue $e\in E$. Queue $e$ is defined as symmetric queue if it operates in the following manner
\begin{enumerate}
	\item The service requirement of a packet is a random variable whose distribution may depend upon the class of the customer.
	\item A total service effort is provided by queue $e$ at the rate $\mu_e(n_e)$.
	\item The packet in position $l$ in the queue is provided with a portion $\gamma_e(l,n_e)$  of the total service effort, for $l=1,2,\cdots,n_e$; when this packet completes service and leaves the queue, packets in positions $l+1,l+2,\cdots,n_e$ move down to positions $l,l+1,\cdots,n_e-1$, respectively.
	\item An arriving packet at queue $e$ moves into position $l$, for $l=1,2,\cdots,n_e$, with probability $\gamma_e(l,n_e+1)$; packets that where in positions $l,l+1,\cdots,n_e+1$, move up to positions $l+1,l+2,\cdots,n_e+1$, respectively.
\end{enumerate} 
Similarly, we require $\mu_e(n_e)>0$ for $n_e>0$; in addition,
\begin{equation}
\sum_{l=1}^{n_e} \gamma_e(l,n_e) = 1,
\end{equation}

As shown in \cite{kelly}, and \cite{nelson}, symmetric queues have product form steady-state distributions. In particular, it turns out the probability of there are $n_e$ packets in queue $e$ is similar to that given by \eqref{eq:steady_gen}. 
\section{Proof of Lemma \ref{lemma:symmetric}}
\label{app:proof-symmetric}
Let $\rho_{(u,v)}(\mathbf{x})$ be the traffic load on queue $(u,v)\in E$, as defined by \eqref{eq:rho_X}. 
It can be shown that the average number of packets in queue $(u,v)\in E$ is of form \cite{datanetworks}
\begin{equation}
\mathbb{E} [n_{(u,v)}(\mathbf{x})]  = \rho_{(u,v)}(\mathbf{x})+\frac{\rho^2_{(u,v)}(\mathbf{x})}{2(1-\rho_{(u,v)}(\mathbf{x}))}.
\end{equation} 
It is easy to see that this function is strictly increasing and convex in $\rho_{(u,v)}(\mathbf{x})$ for $\rho_{(u,v)}(\mathbf{x})\in [0,1)$. Due to Theorem \ref{thm:coststruct}, $N(\mathbf{x}):= \sum_{(u,v)\in E}\mathbb{E} [n_{(u,v)}(\mathbf{x})]$ is  non-increasing and supermodular in $\mathbf{x}$, and the proof is complete.
\section{Proof of Lemma \ref{lemma:example}}
\label{app:md1}
\begin{figure}
\begin{center}
	\includegraphics[scale=0.3]{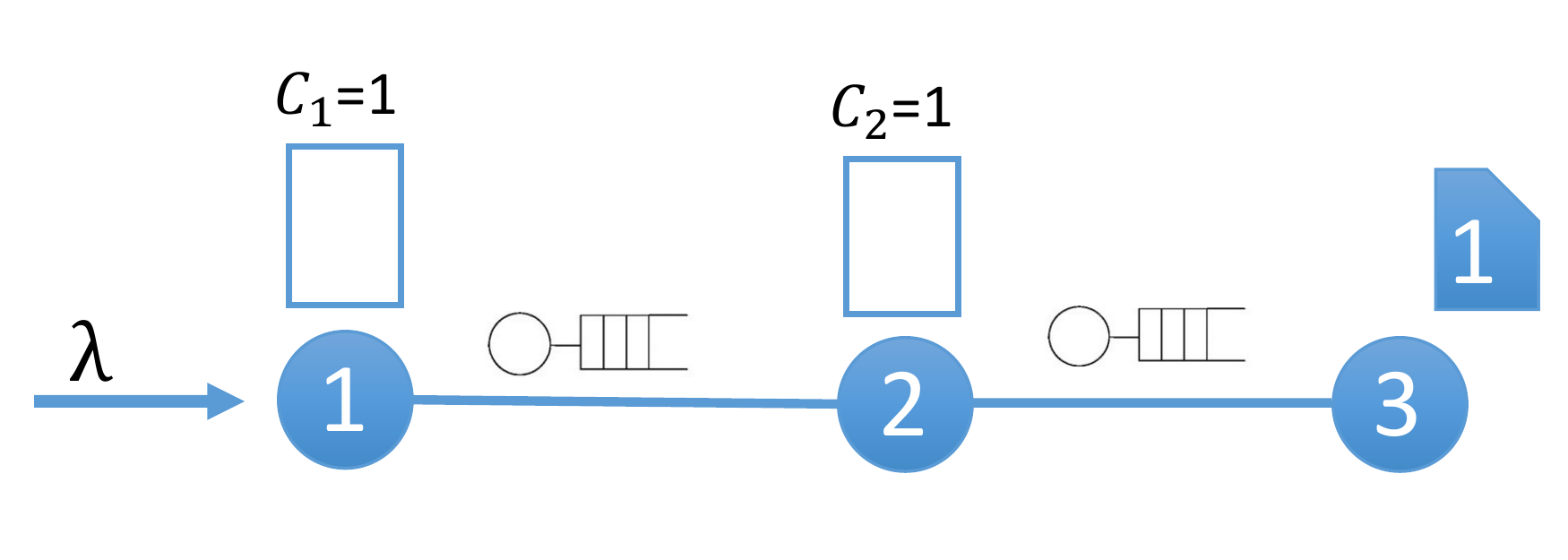}
\end{center}
	\caption{A simple network with finite-capacity queues.}
	\label{fig:finiteQueue}
\end{figure}
\begin{table}[ht]
	\caption{Results of $\rho_{u,v}(\mathbf{x})$'s for different caching configurations.} 
	\label{tab:example} 
	\centering 
	\begin{tabular}{c c c} 
		\hline\hline 
		$[x_{11},x_{21}]$ & $\rho_{3,2}$ & $\rho_{2,1}$ \\ [0.5ex] 
		\hline 
		$[0,0]$ &  $\frac{\lambda }{\mu_{3,2}}$  & $\frac{\lambda(1-p_{3,2}^L)}{\mu_{2,1}}$  \\ [0.5ex] 
		\hline
		$[1,0]$ & 0  & 0 \\[0.5ex] 
		\hline
		$[0,1]$ & 0  &  $\frac{\lambda}{\mu_{2,1}}$ \\[0.5ex] 
		\hline
		$[1,1]$ & 0  &  0 \\[0.5ex] 
		\hline 
	\end{tabular}
\end{table}
Consider the network of $M/M/1/k$ queues in Fig. \ref{fig:finiteQueue}, where node 1 is requesting content 1 from node 3, according to a Poisson process with rate $\lambda$. For simplicity, we only consider the traffic for content 1. For queues $(2,1)$ and $(3,2)$,  it is easy to verify that the probability of packet drop at queues $(u,v) \in \{(2,1),(3,2)\}$ is given by
\begin{eqnarray}
p_{(u,v)}^L(\rho_{(u,v)}) = \frac{\rho_{u,v}(\mathbf{x})^k(1-\rho_{(u,v)}(\mathbf{x}))}{1-\rho_{(u,v)}(\mathbf{x})^{k+1}},
\end{eqnarray}
where $\rho_{(u,v)}(\mathbf{x})$ is the traffic load on queue $(u,v)$, and it can be computed for$(2,1)$ and $(3,2)$ as follows:
\begin{equation}
\rho_{(2,1)}(x_{11},x_{21})= \frac{\lambda (1-x_{11})(1-p_{(3,2)}^L)}{\mu_{(2,1)}},
\end{equation}
\begin{equation}
\rho_{(3,2)}(x_{11},x_{21})= \frac{\lambda (1-x_{11})(1-x_{21})}{\mu_{(3,2)}}.
\end{equation}	

Using the results reported in Table \ref{tab:example}, it is easy to verify that $\rho$'s are not monotone in $\mathbf{x}$. Hence, no strictly monotone function of $\rho$'s are monotone in $\mathbf{x}$. In addition, it can be verified that $\rho$'s are neither submodular, nor supermodular in $\mathbf{x}$. To show this, let sets $A = \emptyset$, and $B = \{(1,1)\}$, correspond to caching configurations $[0,0]$ and $[1,0]$, respectively.  Note that $A \subset B$,  and $(2,1) \notin B$. Since
$
\rho_{(3,2)}(A\cup \{(2,1)\} ) - \rho_{(3,2)}(A) = - \frac{\lambda }{\mu_{(3,2)}} \ngeqslant 0 = \rho_{(3,2)}(B\cup \{(2,1)\} ) - \rho_{(3,2)}(B),
$ 
then $\rho_{(3,2)} $ is not submodular. Consequently, no strictly monotone function of $\rho_{(3,2)} $ is submodular. Similarly, as 
$
\rho_{(2,1)}(A\cup \{(2,1)\} ) - \rho_{(2,1)}(A) = \frac{\lambda p_{(3,2)}^L}{\mu_{(2,1)}} \nleqslant 0 = \rho_{(2,1)}(B\cup \{(2,1)\} ) - \rho_{(2,1)}(B),
$
$\rho_{(2,1)} $ is not supermodular. Thus, no strictly monotone function of $\rho_{(2,1)} $ is supermodular.

}{}

\end{document}